\documentclass[11pt]{article}

\usepackage[utf8]{inputenc}
\usepackage[T1]{fontenc}
\usepackage{multicol}
\usepackage{amsthm}
\usepackage{amsmath}
\usepackage{amssymb}
\usepackage{mathtools}
\usepackage{dsfont}
\usepackage{bm}
\usepackage{bbm}
\usepackage{xparse}
\usepackage{physics}
\usepackage{empheq}
\usepackage{url}
\usepackage{hyperref}
\usepackage[affil-it]{authblk}
\usepackage{tikz}
\usetikzlibrary{shapes,arrows,quotes,angles,calc}

\tikzstyle{block} = [draw, fill=white, rectangle, 
  minimum height=2em, minimum width=3em]
\tikzstyle{bigblock} = [draw, fill=white, rectangle, 
  minimum height=3em, minimum width=4em]
\tikzstyle{Bigblock} = [draw, fill=white, rectangle, 
    minimum height=5em, minimum width=6em]
\tikzstyle{input} = [coordinate]
\tikzstyle{output} = [coordinate]
\tikzstyle{pinstyle} = [pin edge={to-,thin,black}]

\usepackage{rotating}
\usepackage{graphicx}
\usepackage[linesnumbered,ruled,vlined]{algorithm2e}

\usepackage{pgfplots}
\pgfplotsset{compat = newest}

\usepackage[top=3cm,bottom=2cm,right=2cm,left=2cm]{geometry}

\theoremstyle{definition}
\newtheorem{theorem}{Theorem}[section]

\newtheorem{corollary}[theorem]{Corollary}
\newtheorem{proposition}[theorem]{Proposition}
\newtheorem{definition}[theorem]{Definition}

\newtheorem{theorem*}{Theorem}

\theoremstyle{remark}
\newtheorem*{rk}{Remark}

\DeclareMathOperator{\Poi}{\text{Poi}}
\DeclareMathOperator{\Ber}{\text{Ber}}

\DeclareMathOperator{\maxi}{\text{maximize}}

\DeclareMathOperator{\st}{\text{subject to}}
\DeclarePairedDelimiter\ceil{\lceil}{\rceil}

\DeclarePairedDelimiterX\set[1]\lbrace\rbrace{#1}

\title{\bfseries Broadcast Channel Coding: Algorithmic Aspects and Non-Signaling Assistance}
\author{Omar Fawzi\footnote{Univ Lyon, ENS Lyon, UCBL, CNRS, Inria,  LIP, F-69342, Lyon Cedex 07, France. \href{mailto:omar.fawzi@ens-lyon.fr}{\texttt{omar.fawzi@ens-lyon.fr}}} \qquad Paul Fermé\footnote{Univ Lyon, ENS Lyon, UCBL, CNRS, Inria, LIP, F-69342, Lyon Cedex 07, France. \href{mailto:paul.ferme@ens-lyon.fr}{\texttt{paul.ferme@ens-lyon.fr}}}
}
\date{}

\begin{document}
\maketitle

\abstract{We address the problem of coding for classical broadcast channels, which entails maximizing the success probability that can be achieved by sending a fixed number of messages over a broadcast channel. For point-to-point channels, Barman and Fawzi found in~\cite{BF18} a $(1-e^{-1})$-approximation algorithm running in polynomial time, and showed that it is \textrm{NP}-hard to achieve a strictly better approximation ratio. Furthermore, these algorithmic results were at the core of the limitations they established on the power of non-signaling assistance for point-to-point channels. It is natural to ask if similar results hold for broadcast channels, exploiting links between approximation algorithms of the channel coding problem and the non-signaling assisted capacity region.
  
  In this work, we make several contributions on algorithmic aspects and non-signaling assisted capacity regions of broadcast channels. For the class of deterministic broadcast channels, we describe a $(1-e^{-1})^2$-approximation algorithm running in polynomial time, and we show that the capacity region for that class is the same with or without non-signaling assistance. Finally, we show that in the value query model, we cannot achieve a better approximation ratio than $\Omega\left(\frac{1}{\sqrt{m}}\right)$ in polynomial time for the general broadcast channel coding problem, with $m$ the size of one of the outputs of the channel.}

\section{Introduction}
Broadcast channels, introduced by Cover in~\cite{Cover72}, describe the simple network communication setting where one sender aims to transmit individual messages to two receivers. Contrary to point-to-point channels~\cite{Shannon48} or multiple-access channels~\cite{Liao73,Ahlswede73}, the capacity region of broadcast channels is known only for particular classes such as the degraded~\cite{Bergmans73,Gallager74,AK75}, deterministic~\cite{Marton77,Pinsker78} and semi-deterministic~\cite{GIP80}. Only inner bounds~\cite{Cover75,Meulen75,Marton79} and outer bounds~\cite{Sato78,Marton79,NG07,GN20} on the capacity region are known in the general setting.

On the one hand, from the point of view of quantum information, it is natural to ask whether additional resources, such as quantum entanglement or more generally non-signaling correlations between the parties, changes the capacity region. A non-signaling correlation is a multipartite input-output box shared between parties that, as the name suggests, cannot by itself be used to send information between parties. However, non-signaling correlations such as the ones generated by measurements of entangled quantum particles, can provide an advantage for various information processing tasks and nonlocal games. The study of such correlations has given rise to the quantum information area known as nonlocality~\cite{BCPSW14}. For example, in the context of channel coding, there exists classical point-to-point channels for which quantum entanglement between the sender and the receiver can increase the optimal success probability for sending classical information with a single use of the channel~\cite{CLMW10,PLMKR11,BF18}. There are even such channels for which entanglement assistance increases the zero-error capacity~\cite{CLMW10}. However, for classical point-to-point channels, entanglement~\cite{BBCJPW93,BSST99} and even more generally non-signaling correlations~\cite{Matthews12} do not change the capacity of the channel.

In the network setting, the behavior is different. Quek and Shor showed in~\cite{QS17} the existence of two-sender two-receiver interference channels with gaps between their classical, quantum-entanglement assisted and non-signaling assisted capacity regions. Following this result, Leditzky et al.~\cite{LALS20} showed that quantum entanglement shared between the two senders of a multiple access channel can strictly enlarge the capacity region; see~\cite{SLSS22,SLVG23,PDB23} for further results on the effect of entangled transmitters for the multiple access channel. A general investigation of non-signaling resources on multiple-access channel coding was done in~\cite{FF22, FF23}, where it was notably proved that non-signaling advantage occurs even for a simple textbook multiple-access channel: the binary adder channel. However, the influence of nonlocal resources on broadcast channels has been comparably less studied. We only know that quantum entanglement shared between decoders does not change the capacity region~\cite{PDB21}. In a different context, quantum illumination~\cite{Llo08} shows that entanglement between a signal photon and an idler photon can enhance the performance of sensing.

On the other hand, from an algorithmic point of view, an important question is the complexity of the channel coding problem, which entails maximizing the success probability that can be achieved by sending a fixed number of messages over a channel. However, as solving exactly this problem is \textrm{NP}-hard, a natural question that arises is its approximability. For point-to-point channels, Barman and Fawzi found in~\cite{BF18} a $(1-e^{-1})$-approximation algorithm running in polynomial time. They showed that it is \textrm{NP}-hard to approximate the channel coding problem in polynomial time for any strictly better ratio. For $\ell$-list-decoding, where the decoder is allowed to output a list of $\ell$ guesses, a polynomial-time approximation algorithm achieving a $1-\frac{\ell^{\ell}e^{-\ell}}{\ell!}$ ratio was found in~\cite{BFGG20}, and it was shown to be \text{NP}-hard to do better in \cite{BFF21}. For multiple-access channel coding, the complexity of the problem can be linked to the bipartite densest subgraph problem~\cite{FKP01, Ferme23}, which cannot be approximated within any constant ratio under a complexity hypothesis on random $k$-SAT formulas~\cite{AAMMW11}. However, the approximability of broadcast channel coding has not been addressed in the literature.

In the point-to-point scenario studied in~\cite{BF18}, the existence of a constant-ratio approximation algorithm is linked to the equality of the capacity regions with and without non-signaling assistance. Indeed, giving non-signaling assistance to the channel coding problem turns it into a linear program, thus computable in polynomial time. In fact, it is equal to its natural linear relaxation, which is a common strategy towards approximating an integer linear program. Showing that this approximation strategy guarantees a constant ratio is the key ingredient in proving the equality of the capacity regions with and without non-signaling assistance. 
This raises the following questions on broadcast channels: Does the capacity region of the broadcast channel change when non-signaling resources between parties are allowed? What is the best approximability ratio of the broadcast channel coding problem? How those two questions are related?

\paragraph{Contributions} 
As a first result, we prove that the sum success probabilities of the broadcast channel coding problem are the same with and without non-signaling assistance between decoders; see Theorem~\ref{theo:NSdecoders}. This strengthens a result by~\cite{PDB21} establishing that entanglement between the decoders does not change the capacity region.

The main focus of this paper is to study the influence of sharing a non-signaling resource between the three parties. Our main result shows that for the class of deterministic broadcast channels, non-signaling resources shared between the three parties does not change the capacity region; see Theorem~\ref{theo:NSdet} and Corollary~\ref{cor:NSdet}. In order to prove this result, we consider the algorithmic problem of optimal channel coding for a deterministic broadcast channel. For this problem, we describe a $(1-e^{-1})^2$-approximation algorithm running in polynomial time. This is achieved through a graph interpretation of the problem, where one aims at partitioning a bipartite graph into $k_1$ and $k_2$ parts, such that the resulting quotient graph is the densest possible; see Proposition~\ref{prop:BCCisDQG} and Theorem~\ref{theo:DQGapprox}. To prove our result on the limitations of non-signaling assistance for deterministic broadcast channels, we use the same ideas as the ones involved in the analysis of the approximation algorithm.

As far as hardness is concerned, we consider the subproblem of broadcast channel coding where the number of messages one decoder is responsible of is maximum. This subproblem can be interpreted as a social welfare maximization problem. In the theory of fair division~\cite{BT96,Moulin03}, social welfare maximization entails partitioning a set of goods among agents in order to maximize the sum of their utilities. The social welfare problem has been extensively studied through a black box approach~\cite{BN05}, which led to a precise analysis of achievable approximation ratio as well as hardness results~\cite{DS06,MSV08}, depending on the class of utility functions considered and the type of black box access to them. We refine the hardness result for the class of fractionally sub-additive utility functions to the subclass coming from the broadcast channel coding subproblem interpretation. Specifically, we show that in the value query model, we cannot achieve a better approximation ratio than $\Omega\left(\frac{1}{\sqrt{m}}\right)$ in polynomial time, with $m$ the size of one of the outputs of the channel: see Theorem \ref{theo:VQhardnessBC}. This gives some evidence that the broadcast channel coding problem might be hard to approximate. Following the previous discussion on the links between approximation algorithms and non-signaling capacity regions, this hardness evidence is a first step towards showing that sharing a non-signaling resource between the three parties of a broadcast channel can enlarge its capacity region.

\paragraph{Organization} In Section~\ref{section:prelim}, we introduce some basic definitions as well as useful notions that will be used throughout this work. In Section~\ref{section:BCCpb}, we define precisely the different versions of the broadcast channel coding problem depending on the choice of objective value, and show that they all lead to the same capacity region. In Section~\ref{section:NSBC}, we define the different non-signaling assisted versions of the broadcast channel coding problem. In particular, we show that the sum success probabilities with and without non-signaling assistance shared between decoders are the same, and that it implies that the related capacity regions are the same. In Section~\ref{section:ApproxDetBC}, we address both algorithmic aspects and capacity considerations of deterministic broadcast channels. Specifically, we describe a $(1-e^{-1})^2$-approximation algorithm running in polynomial time for that class, and we show that the capacity region for that class is the same with or without non-signaling assistance. Finally, in Section~\ref{section:HardnessBC}, we show that in the value query model, we cannot achieve a better approximation ratio than $\Omega\left(\frac{1}{\sqrt{m}}\right)$ in polynomial time for the general broadcast channel coding problem, with $m$ the size of one of the outputs of the channel.

\section{Preliminaries}
\label{section:prelim}
\subsection{Broadcast Channels}
Formally, a broadcast channel is given by a conditional probability distribution on input $\mathcal{X}$ and two outputs $\mathcal{Y}_1$ and $\mathcal{Y}_2$, so $W := \left(W(y_1,y_2|x)\right)_{y_1 \in \mathcal{Y}_1, y_2 \in \mathcal{Y}_2, x \in \mathcal{X}}$, with $W(y_1y_2|x) \geq 0$ and such that $\sum_{y_1 \in \mathcal{Y}_1, y_2 \in \mathcal{Y}_2} W(y_1y_2|x) = 1$. We define its marginals $W_1$ and $W_2$ respectively by $W_1(y_1|x) := \sum_{y_2 \in \mathcal{Y}_2} W(y_1y_2|x)$ and $W_2(y_2|x) := \sum_{y_1 \in \mathcal{Y}_1} W(y_1y_2|x)$. We will denote such a broadcast channel by $W : \mathcal{X} \rightarrow \mathcal{Y}_1 \times \mathcal{Y}_2$. The tensor product of two broadcast channels $W: \mathcal{X} \rightarrow \mathcal{Y}_1 \times \mathcal{Y}_2$ and $W': \mathcal{X}' \rightarrow \mathcal{Y}'_1 \times \mathcal{Y}'_2$ is denoted by $W \otimes W' : \mathcal{X} \times \mathcal{X}' \rightarrow (\mathcal{Y}_1 \times \mathcal{Y}_1') \times (\mathcal{Y}_2 \times \mathcal{Y}_2')$ and defined by $(W \otimes W')(y_1y_1'y_2y_2'|xx') := W(y_1y_2|x) \cdot W'(y_1'y_2'|x')$. We define $W^{\otimes n}(y_1^ny_2^n|x^n) := \prod_{i=1}^nW(y_{1,i}y_{2,i}|x_i)$, for $y_1^n := y_{1,1} \ldots y_{1,n} \in \mathcal{Y}_1^n$ and $y_2^n := y_{2,1} \ldots y_{2,n} \in \mathcal{Y}_2^n$ and $x^n := x_1 \ldots x_n \in \mathcal{X}^n$. We will use the notation $[k]:=\{1,\ldots,k\}$.

\subsection{Capacity Regions}
Given a notion of success probability $\mathrm{S}(W,k_1,k_2)$, that is to say the probability of correctly encoding and decoding $k_1$ and $k_2$ messages for the channel $W$, we can define the related capacity region. 
\begin{definition}[Capacity Region ${\mathcal{C}[\mathrm{S}](W)}$ for a success probability $\mathrm{S}(W,k_1,k_2)$]
  \label{defi:generalCapacityRegion}
  A rate pair $(R_1,R_2)$ is $\mathrm{S}$-achievable (for the channel $W$) if:
  \[ \underset{n \rightarrow +\infty}{\lim} \mathrm{S}(W^{\otimes n},\ceil{2^{R_1n}},\ceil{2^{R_2n}}) = 1 \ . \]
  We define the $\mathrm{S}$-capacity region $\mathcal{C}[\mathrm{S}](W)$ as the closure of the set of all achievable rate pairs (for the channel $W$).
\end{definition}


\subsection{Negatively Associated Random Variables}
We present a weaker notion of independence for random variables which is called negative association as introduced in~\cite{JP83}, for which the Chernoff-Hoeffding bounds still hold.

\begin{definition}
   \label{defi:NA}
Random variables $X_1,\ldots,X_n$ are said to be negatively associated if for every pair of disjoints subsets $I,J$ of $[n]$ and (coordinate-wise) increasing functions $f,g$, we have:
    \[ \mathbb{E}[f(\{X_i : i \in I\})) \cdot g(\{X_i : i \in J\})] \leq  \mathbb{E}[f(\{X_i : i \in I\})] \cdot \mathbb{E}[g(\{X_i : i \in J\})] \ .\]
\end{definition}

\begin{proposition}[Property $\text{P}_1$ of~\cite{JP83}]
  \label{prop:NA_P1}
  A pair of random variable $X,Y$ is negatively associated if and only if:
  \[ \forall x \in \mathcal{X},  \forall y \in \mathcal{Y}, P_{XY}(x,y) \leq P_X(x)P_Y(y) \ .\]
\end{proposition}

\begin{proposition}[Property $\text{P}_4$ of~\cite{JP83}]
  \label{prop:NA_P4}
  A subset of two or more negatively associated random variables is negatively associated.
\end{proposition}

\begin{proposition}[Property $\text{P}_5$ of~\cite{JP83}]
  \label{prop:NA_P5}
  A set of independent random variables is negatively associated.
\end{proposition}

\begin{proposition}[Property $\text{P}_6$ of~\cite{JP83}]
  \label{prop:NA_P6}
  Increasing functions defined on disjoint subsets of a set of negatively associated random variables are negatively associated.
\end{proposition}

\begin{proposition}[Property $\text{P}_7$ of~\cite{JP83}]
    \label{prop:NA_P7}
    The union of independent sets of negatively associated random variables is negatively associated.
\end{proposition}

\begin{definition}[Permutation Distribution]
  \label{defi:perm}
  Let $x = (x_1,\ldots,x_k) \in \mathbb{R}^k$. A permutation distribution is the joint distribution of the vector $X = (X_1,\ldots,X_k)$ which takes as values all $k!$ permutations of $x$ with equal probabilities, each being $\frac{1}{k!}$.
\end{definition}

\begin{proposition}[Theorem 2.11 of~\cite{JP83}]
  \label{prop:permNA}
   A permutation distribution is negatively associated.
\end{proposition} 

\begin{proposition}[Chernoff-Hoeffding bound]
  \label{prop:chernoff2}
  Let $X_1, \ldots, X_n$ be negatively associated Bernouilli random variables of parameter $p$. Then for $0 < \varepsilon \leq \frac{1}{2}$:
  \[ \mathbb{P}\left(\frac{1}{n}\sum_{i=1}^n X_i > (1+\varepsilon)p \right) \leq e^{-\frac{pn\varepsilon^2}{4}} \ . \]
\end{proposition}
\begin{proof}
  Usual proofs of the Chernoff-Hoeffding bound work in the same way with negatively associated variables as pointed out by~\cite{DR98}. So, one obtain as in the original proof (Theorem 1 of~\cite{Hoeffding63}) that:
  \[ \mathbb{P}\left(\frac{1}{n}\sum_{i=1}^n X_i > (1+\varepsilon)p \right) \leq e^{-D\left((1+\varepsilon)p||p\right)n} \ ,\]
    with $D\left(x||y\right) := x\ln\left(\frac{x}{y}\right) + (1-x)\ln\left(\frac{1-x}{1-y}\right)$ the Kullback–Leibler divergence between Bernoulli distributed random variables with parameters $x$ and $y$. As $D\left((1+\varepsilon)p||p\right) \geq \frac{\varepsilon^2p}{4}$ for $0 < \varepsilon < \frac{1}{2}$, we recover the expected bound.
\end{proof}

\section{Broadcast Channel Coding}
\label{section:BCCpb}
\subsection{Broadcast Channels}
The coding problem for a broadcast channel $W : \mathcal{X} \rightarrow \mathcal{Y}_1 \times \mathcal{Y}_2$ can be stated in the following way. We want to encode a pair of messages belonging to $[k_1] \times [k_2]$ into $\mathcal{X}$. The pair is given as input to $W$, which results in two random outputs in $\mathcal{Y}_1$ and $\mathcal{Y}_2$. From the output in $\mathcal{Y}_1$ (resp. $\mathcal{Y}_2$), we want to decode back the original message in $[k_1]$ (resp. $[k_2]$). We will call $e : [k_1] \times [k_2] \rightarrow \mathcal{X}$ the encoder, $d_1 : \mathcal{Y}_1 \rightarrow [k_1]$ the first decoder and $d_2 : \mathcal{Y}_2 \rightarrow [k_2]$ the second decoder. The scenario is depicted in Figure~\ref{fig:BCcoding}.

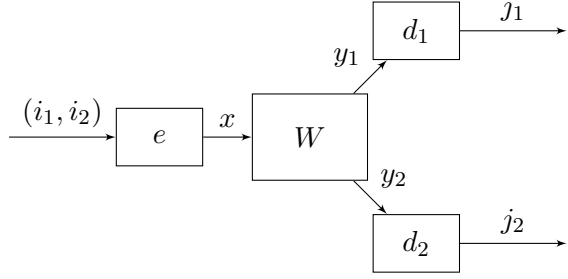
\begin{figure}[!h]
\begin{center}
  \begin{tikzpicture}[auto, node distance=2cm, >=latex']
    \node [input, name=i1i2] {};
    \node [block, right of=i1i2] (e) {$e$};
    \node [bigblock, right of=e] (W) {$W$};
    \node [block, above right of=W] (d1) {$d_1$};
    \node [block, below right of=W] (d2) {$d_2$};

    \draw [->] (e) -- node[name=x] {$x$} (W);
    \draw [->] (W) -- node[name=y1] {$y_1$} (d1);
    \draw [->] (W) -- node[name=y2] {$y_2$} (d2);
    \node [output, right of=d1] (j1) {};
    \node [output, right of=d2] (j2) {};

    \draw [draw,->] (i1i2) -- node {$(i_1,i_2)$} (e);
    \draw [draw,->] (d1) -- node {$j_1$} (j1);
    \draw [draw,->] (d2) -- node {$j_2$} (j2);
  \end{tikzpicture}
\end{center}
\caption{Coding for a broadcast channel $W$.}
\label{fig:BCcoding}
\end{figure}

We will call $\mathrm{p}_1(W,e,d_1)$ (resp. $\mathrm{p}_2(W,e,d_2)$) the probability of successfully decoding the first (resp. second) message, i.e. that $j_1 = i_1$ (resp. $j_2 = i_2$), given that the encoder is $e$ and the decoder is $d_1$ (resp. $d_2$). We will also consider $\mathrm{p}(W,e,d_1,d_2)$, the probability of successfully decoding both messages, i.e. that $j_1 = i_1$ and  $j_2 = i_2$, given that the encoder is $e$ and the decoders are $d_1,d_2$.

We aim to find the best encoder and decoders according to some figure of merit. However, to do so, we need a one-dimensional real-valued objective to optimize. This leads to two different quantities.

\subsection{The Sum Success Probability $\mathrm{S}_{\text{sum}}(W,k_1,k_2)$}
We will focus first on maximizing $\frac{\mathrm{p}_1(W,e,d_1)+\mathrm{p}_2(W,e,d_2)}{2}$ over all encoders $e$ and decoders $d_1,d_2$. We will call $\mathrm{S}_{\text{sum}}(W,k_1,k_2)$ the resulting maximum sum probability of successfully encoding and decoding the messages through $W$, given that the input pair of messages is taken uniformly in $[k_1] \times [k_2]$. $\mathrm{S}_{\text{sum}}(W,k_1,k_2)$ is the solution of the following optimization program:
\begin{equation}
  \begin{aligned}
    &&\underset{e,d_1,d_2}{\maxi} &&& \frac{1}{k_1k_2}\sum_{i_1,i_2,x,y_1,y_2} W(y_1y_2|x)e(x|i_1i_2)\frac{d_1(i_1|y_1) + d_2(i_2|y_2)}{2}\\
    &&\st &&& \sum_{x \in \mathcal{X}} e(x|i_1i_2) = 1, \forall i_1 \in [k_1], i_2 \in [k_2]\\
    &&&&& \sum_{j_1 \in [k_1]} d_1(j_1|y_1) = 1, \forall y_1 \in \mathcal{Y}_1\\
    &&&&& \sum_{j_2 \in [k_2]} d_2(j_2|y_2) = 1, \forall y_2 \in \mathcal{Y}_2\\
    &&&&& e(x|i_1i_2), d_1(j_1|y_1), d_2(j_2|y_2) \geq 0
  \end{aligned}
\end{equation}

\begin{proof}
One should note that we allow in fact non-deterministic encoders and decoders for generality reasons, although the value of the program is not changed as it is convex. Besides that remark, let us name $I_1,I_2,J_1,J_2,X,Y_1,Y_2$ the random variables corresponding to $i_1,i_2,j_1,j_2,x,y_1,y_2$ in the coding and decoding process. Then, given $e,d_1,d_2$ and $W$, the objective value of the previous program comes from:
\begin{equation}
  \begin{aligned}
    &\mathrm{p}_1(W,e,d_1) = \mathbb{P}\left(J_1 = I_1\right) = \frac{1}{k_1k_2}\sum_{i_1,i_2} \mathbb{P}\left(J_1 = i_1|I_1=i_1,I_2=i_2\right)\\
    &= \frac{1}{k_1k_2}\sum_{i_1,i_2,x}e(x|i_1i_2) \mathbb{P}\left(J_1 = i_1|i_1,i_2,X=x\right)\\
    &= \frac{1}{k_1k_2}\sum_{i_1,i_2,x,y_1,y_2}W(y_1y_2|x)e(x|i_1i_2) \mathbb{P}\left(J_1 = i_1|i_1,i_2,x,Y_1=y_1,Y_2=y_2\right)\\
    &= \frac{1}{k_1k_2}\sum_{i_1,i_2,x,y_1,y_2}W(y_1y_2|x)e(x|i_1i_2)d_1(i_1|y_1) \ ,
  \end{aligned}
\end{equation}
and symmetrically for $\mathrm{p}_2(W,e,d_2)$, which leads to the announced objective value.
\end{proof}

One can rewrite this optimization program in a more convenient way, proving that $\mathrm{S}_{\text{sum}}(W,k_1,k_2)$ depends only on the marginals of $W$:
\begin{proposition}
  \begin{equation}
  \begin{aligned}
    \mathrm{S}_{\text{sum}}(W,k_1,k_2) = &&\underset{e,d_1,d_2}{\maxi} &&& \frac{1}{2k_1k_2}\sum_{i_1,x,y_1} W_1(y_1|x)d_1(i_1|y_1)\sum_{i_2} e(x|i_1i_2)\\
    &&+&&& \frac{1}{2k_1k_2}\sum_{i_2,x,y_2} W_2(y_2|x)d_2(i_2|y_2)\sum_{i_1} e(x|i_1i_2)\\
    &&\st &&& \sum_{x \in \mathcal{X}} e(x|i_1i_2) = 1, \forall i_1 \in [k_1], i_2 \in [k_2]\\
    &&&&& \sum_{j_1 \in [k_1]} d_1(j_1|y_1) = 1, \forall y_1 \in \mathcal{Y}_1\\
    &&&&& \sum_{j_2 \in [k_2]} d_2(j_2|y_2) = 1, \forall y_2 \in \mathcal{Y}_2\\
    &&&&& e(x|i_1i_2), d_1(j_1|y_1), d_2(j_2|y_2) \geq 0
  \end{aligned}
  \end{equation}
\end{proposition}

\begin{proof}
  It follows from the definitions $W_1(y_1|x) := \sum_{y_2} W(y_1y_2|x)$ and $W_2(y_2|x) := \sum_{y_1} W(y_1y_2|x)$.
\end{proof}

\subsection{The Joint Success Probability $\mathrm{S}(W,k_1,k_2)$}
We will now focus on maximizing $\mathrm{p}(W,e,d_1,d_2)$ over all encoders $e$ and decoders $d_1,d_2$. We will call $\mathrm{S}(W,k_1,k_2)$ the resulting maximum probability of successfully encoding and decoding the messages through $W$, given that the input pair of messages is taken uniformly in $[k_1] \times [k_2]$. $\mathrm{S}(W,k_1,k_2)$ is the solution of the following optimization program:
\begin{equation}
  \begin{aligned}
    \mathrm{S}(W,k_1,k_2) := &&\underset{e,d_1,d_2}{\maxi} &&& \frac{1}{k_1k_2}\sum_{i_1,i_2,x,y_1,y_2} W(y_1y_2|x)e(x|i_1i_2)d_1(i_1|y_1)d_2(i_2|y_2)\\
    &&\st &&& \sum_{x \in \mathcal{X}} e(x|i_1i_2) = 1, \forall i_1 \in [k_1], i_2 \in [k_2]\\
    &&&&& \sum_{j_1 \in [k_1]} d_1(j_1|y_1) = 1, \forall y_1 \in \mathcal{Y}_1\\
    &&&&& \sum_{j_2 \in [k_2]} d_2(j_2|y_2) = 1, \forall y_2 \in \mathcal{Y}_2\\
    &&&&& e(x|i_1i_2), d_1(j_1|y_1), d_2(j_2|y_2) \geq 0
  \end{aligned}
\end{equation}

The proof is the same as in the sum probability scenario. We define the (resp. sum) capacity region using Definition~\ref{defi:generalCapacityRegion} by $\mathcal{C}(W):=\mathcal{C}[\mathrm{S}](W)$ (resp. $\mathcal{C}_{\text{sum}}(W):=\mathcal{C}[{\mathrm{S}}_{\text{sum}}](W)$).

The objective values of those two optimization programs are not the same, but $\mathrm{S}(W,k_1,k_2)$ and $\mathrm{S}_{\text{sum}}(W,k_1,k_2)$ still characterize the same capacity region~\cite{Willems90}: 


\begin{proposition}
  \label{prop:capacitysum}
  For any broadcast channel $W$, $\mathcal{C}(W) = \mathcal{C}_{\text{sum}}(W)$.
\end{proposition}
\begin{proof}
  Let us focus on error probabilities rather than success ones. Call them respectively $\mathrm{E}(W,k_1,k_2) := 1-\mathrm{S}(W,k_1,k_2)$ and $\mathrm{E}_{\text{sum}}(W,k_1,k_2) := 1-\mathrm{S}_{\text{sum}}(W,k_1,k_2)$. Let us fix a solution $e,d_1,d_2$ of the optimization program computing $\mathrm{S}(W,k_1,k_2)$. Let us remark first that:
  \[ \sum_{i_1,i_2,x,y_1,y_2} W(y_1y_2|x)e(x|i_1i_2) = k_1k_2\ , \]
  thus, the value of the maximum error for those encoder and decoders is:
\begin{equation}
  \begin{aligned}
    &\mathrm{E}(W,k_1,k_2,e,d_1,d_2) := 1 -  \frac{1}{k_1k_2}\left(\sum_{i_1,i_2,x,y_1,y_2} W(y_1y_2|x)e(x|i_1i_2)d_1(i_1|y_1)d_2(i_2|y_2)\right)\\
    &=\frac{1}{k_1k_2}\sum_{i_1,i_2,x,y_1,y_2} W(y_1y_2|x)e(x|i_1i_2)\\
    &-\frac{1}{k_1k_2}\sum_{i_1,i_2,x,y_1,y_2} W(y_1y_2|x)e(x|i_1i_2)d_1(i_1|y_1)d_2(i_2|y_2)\\
    &=\frac{1}{k_1k_2}\left(\sum_{i_1,i_2,x,y_1,y_2} W(y_1y_2|x)e(x|i_1i_2)\left[1-d_1(i_1|y_1)d_2(i_2|y_2)\right]\right) \ .\\
  \end{aligned}
\end{equation}

Similarly, the value of the sum error $\mathrm{E}_{\text{sum}}(W,k_1,k_2,e,d_1,d_2)$ is equal to:
\begin{equation}
  \begin{aligned}
    &1 -  \frac{1}{k_1k_2}\left(\sum_{i_1,i_2,x,y_1,y_2} W(y_1y_2|x)e(x|i_1i_2)\frac{d_1(i_1|y_1)+d_2(i_2|y_2)}{2}\right)\\
    &= \frac{1}{k_1k_2}\left(\sum_{i_1,i_2,x,y_1,y_2} W(y_1y_2)e(x|i_1i_2)\left[1-\frac{d_1(i_1|y_1)+d_2(i_2|y_2)}{2}\right]\right) \ .
  \end{aligned}
\end{equation}

However, for $x,y \in [0,1]$, we have that:
\[1-xy \geq \max\left(1-x,1-y\right) \geq 1-\frac{x+y}{2} \ , \]
and:
\[1-xy \leq (1-x) + (1-y) = 2  \left(1-\frac{x+y}{2}\right) \ . \]

This means that, for all $e,d_1,d_2$, we have:
\[ \mathrm{E}_{\text{sum}}(W,k_1,k_2,e,d_1,d_2) \leq \mathrm{E}(W,k_1,k_2,e,d_1,d_2) \leq 2\mathrm{E}_{\text{sum}}(W,k_1,k_2,e,d_1,d_2) \ ,\]
so, maximizing over all $e,d_1,d_2$, we get:
\[ \mathrm{E}_{\text{sum}}(W,k_1,k_2) \leq \mathrm{E}(W,k_1,k_2) \leq 2\mathrm{E}_{\text{sum}}(W,k_1,k_2) \ .\]

Thus, up to a multiplicative factor $2$, the error is the same. In particular, when one of those errors tends to zero, the other one tends to zero as well. This implies that the capacity regions are the same.
\end{proof}

\section{Non-Signaling Assistance}
\label{section:NSBC}
In this section, we will consider the broadcast channel coding problem with additional resources, in order to determine how these resources affect its success probabilities as well as the capacity regions that can be defined from them.

\subsection{Non-Signaling Assistance Between Decoders}
Here, we consider the case where the receivers are given non-signaling assistance. This resource, which is a theoretical but easier to manipulate generalization of quantum entanglement, can be characterized as follows.  A non-signaling box $d(j_1j_2|y_1y_2)$ is any joint conditional probability distribution such that the marginal from one party is independent of the other party's input, i.e. we have:
\begin{equation}
  \begin{aligned}
    &&\forall j_1,y_1,y_2,y'_2, &&&\sum_{j_2} d(j_1j_2|y_1y_2) = \sum_{j_1} d(j_1j_2|y_1y'_2) \ ,\\
    &&\forall j_2,y_1,y_2,y'_1, &&&\sum_{j_1} d(j_1j_2|y_1y_2) = \sum_{j_1} d(j_1j_2|y'_1y_2) \ .
  \end{aligned}
\end{equation}

Thus, when receivers are given non-signaling assistance, the product $d_1(j_1|y_1)d_2(j_2|y_2)$ is replaced by the non-signaling box $d(j_1j_2|y_1y_2)$. Thus, we define the joint and sum success probabilities $\mathrm{S}^{\mathrm{NS}_{\text{dec}}}(W,k_1,k_2)$ (resp. $\mathrm{S}_{\text{sum}}^{\mathrm{NS}_{\text{dec}}}(W,k_1,k_2)$) by:
\begin{equation}
  \begin{aligned}
    &&\underset{e,d_1,d_2}{\maxi} &&& \frac{1}{k_1k_2}\sum_{i_1,i_2,x,y_1,y_2} W(y_1y_2|x)e(x|i_1i_2)d(i_1i_2|y_1y_2)\\
    \Big(\text{resp. }&&\underset{e,d_1,d_2}{\maxi} &&& \frac{1}{2k_1k_2}\sum_{i_1,i_2,x,y_1,y_2} W(y_1y_2|x)e(x|i_1i_2)\sum_{j_2}d(i_1j_2|y_1y_2)\\
    &&+&&& \frac{1}{2k_1k_2}\sum_{i_1,i_2,x,y_1,y_2} W(y_1y_2|x)e(x|i_1i_2)\sum_{j_1}d(j_1i_2|y_1y_2)\Big)\\
    &&\st &&& \sum_x e(x|i_1i_2) = 1\\
    &&&&& \sum_{j_2} d(j_1j_2|y_1y_2) = \sum_{j_1} d(j_1j_2|y_1y'_2)\\
    &&&&& \sum_{j_1} d(j_1j_2|y_1y_2) = \sum_{j_1} d(j_1j_2|y'_1y_2)\\
    &&&&& \sum_{j_1,j_2} d(j_1j_2|y_1y_2) = 1\\
    &&&&& e(x|i_1i_2), d(j_1j_2|y_1y_2) \geq 0\\
  \end{aligned}
\end{equation}

The (resp. sum) capacity region with non-signaling assistance between decoders is defined using Definition~\ref{defi:generalCapacityRegion} by $\mathcal{C}^{\mathrm{NS}_{\text{dec}}}(W) := \mathcal{C}[\mathrm{S}^{\mathrm{NS}_{\text{dec}}}](W)$ (resp. $\mathcal{C}^{\mathrm{NS}_{\text{dec}}}_{\text{sum}}(W) := \mathcal{C}[\mathrm{S}_{\text{sum}}^{\mathrm{NS}_{\text{dec}}}](W)$).



We will now show that sum and joint capacity regions with non-signaling assistance between decoders are the same.

\begin{proposition}
  \label{prop:NSdecoderssum}
  For any broadcast channel $W$, $\mathcal{C}_{\text{sum}}^{\mathrm{NS}_{\text{dec}}}(W)= \mathcal{C}^{\mathrm{NS}_{\text{dec}}}(W)$.
\end{proposition}

\begin{proof}
  Given an encoder $e$ and a non-signaling decoding box $d$, the maximum success probability of encoding and decoding correctly with those is given by:
\[ \mathrm{S}^{\mathrm{NS}_{\text{dec}}}(W,k_1,k_2,e,d) := \frac{1}{k_1k_2}\sum_{i_1,i_2,x,y_1,y_2} W(y_1y_2|x)e(x|i_1i_2)d(i_1i_2|y_1y_2) \ . \]

This should be compared to the sum success probability $\mathrm{S}_{\text{sum}}^{\mathrm{NS}_{\text{dec}}}(W,k_1,k_2,e,d)$ of encoding and decoding correctly with those:
\[ \frac{1}{k_1k_2}\sum_{i_1,i_2,x,y_1,y_2} W(y_1y_2|x)e(x|i_1i_2)\left[\frac{\sum_{j_2}d(i_1j_2|y_1y_2)+\sum_{j_1}d(j_1i_2|y_1y_2)}{2}\right] \ . \]

Similarly to what was done in Proposition~\ref{prop:capacitysum}, we focus on error probabilities rather than success probabilities. This leads again to:
\[ \mathrm{E}^{\mathrm{NS}_{\text{dec}}}(W,k_1,k_2,e,d) = \frac{1}{k_1k_2}\sum_{i_1,i_2,x,y_1,y_2} W(y_1y_2|x)e(x|i_1i_2)\left[1-d(i_1i_2|y_1y_2)\right] \ , \]
and $\mathrm{E}_{\text{sum}}^{\mathrm{NS}_{\text{dec}}}(W,k_1,k_2,e,d)$ equal to:
\[ \frac{1}{k_1k_2}\sum_{i_1,i_2,x,y_1,y_2} W(y_1y_2|x)e(x|i_1i_2)\left[\frac{1-\sum_{j_2}d(i_1j_2|y_1y_2)}{2}+\frac{1-\sum_{j_1}d(j_1i_2|y_1y_2)}{2}\right] \ . \]

But we have that:
\begin{equation}
  \begin{aligned}
    1-d(i_1i_2|y_1y_2) &\geq \max\left(1-\sum_{j_2}d(i_1j_2|y_1y_2),1-\sum_{j_1}d(j_1i_2|y_1y_2)\right)\\
    &\geq \frac{1-\sum_{j_2}d(i_1j_2|y_1y_2)}{2}+\frac{1-\sum_{j_1}d(j_1i_2|y_1y_2)}{2} \ ,
  \end{aligned}
\end{equation}
since $d(j_1j_2|y_1y_2) \in [0,1]$, and we have that: 
\begin{equation}
  \begin{aligned}
    &1-\sum_{j_2}d(i_1j_2|y_1y_2)+1-\sum_{j_1}d(j_1i_2|y_1y_2)\\
    &= 1-d(i_1i_2|y_1y_2) + 1-\sum_{(j_1,j_2) \in S}d(j_1j_2|y_1y_2)\\
    &\geq 1-d(i_1i_2|y_1y_2) \ ,
  \end{aligned}
\end{equation}
with $S := \{(i_1,j_2): j_2 \in [k_2]-\{i_2\}\} \sqcup \{(j_1,i_2): j_1 \in [k_1]-\{i_1\}\}$.

Thus, this implies that:
\[\mathrm{E}_{\text{sum}}^{\mathrm{NS}_{\text{dec}}}(W,k_1,k_2,e,d) \leq \mathrm{E}^{\mathrm{NS}_{\text{dec}}}(W,k_1,k_2,e,d) \leq 2\mathrm{E}_{\text{sum}}^{\mathrm{NS}_{\text{dec}}}(W,k_1,k_2,e,d) \ ,\]
and by maximizing over all $e$ and $d$:
\[\mathrm{E}_{\text{sum}}^{\mathrm{NS}_{\text{dec}}}(W,k_1,k_2) \leq \mathrm{E}^{\mathrm{NS}_{\text{dec}}}(W,k_1,k_2) \leq 2\mathrm{E}_{\text{sum}}^{\mathrm{NS}_{\text{dec}}}(W,k_1,k_2) \ .\]
As before, this implies that the capacity regions are the same.
\end{proof}

We will now prove that sum success probabilities of the broadcast channel coding problem are the same with and without non-signaling assistance between decoders. In particular, this implies that non-signaling resources shared between decoders does not change the capacity region. Note that, after the publication of~\cite{PDB21}, Pereg et al. added a remark to the arXiv version of their paper that their result stating that entanglement shared between decoders does not change the capacity of a broadcast channel could be generalized to non-signaling assistance. The theorem below strengthens this result showing that non-signaling assistance between the decoders cannot increase the sum success probability even in the one-shot setting and for arbitrary broadcast channels. In particular, this result shows that non-signaling assistance between the decoders does not increase the zero-error capacities nor does it change the error exponents. This is in contrast to correlations between the sender and the receiver in the point-to-point setting that can enhance the zero-error capacity~\cite{CLMW10}. The reason behind this result is very simple: as the quantity $\mathrm{S}_{\text{sum}}(W,k_1,k_2)$ only depends on the marginal channels $W_1$ and $W_2$, correlations between the receivers do not have any effect.

\begin{theorem}
   \label{theo:NSdecoders}
   For any broadcast channel $W$ and $k_1, k_2$, we have  $\mathrm{S}_{\text{sum}}(W,k_1,k_2)=\mathrm{S}_{\text{sum}}^{\mathrm{NS}_{\text{dec}}}(W,k_1,k_2)$. As a consequence, $\mathcal{C}(W) = \mathcal{C}^{\mathrm{NS}_{\text{dec}}}(W)$.
\end{theorem}

\begin{proof}
In the sum scenario, since the objective function does not depend on the product $d_1(j_1|y_1)d_2(j_2|y_2)$ but only on the marginals $d_1(j_1|y_1)$ and $d_2(j_2|y_2)$, the non-signaling box won't give additional decoding power. Indeed, for any encoder $e$ and non-signaling decoding box $d$, we have that:
\begin{equation}
  \begin{aligned}
    \mathrm{S}_{\text{sum}}^{\mathrm{NS}_{\text{dec}}}(W,k_1,k_2,e,d) &:= \frac{1}{2k_1k_2}\sum_{i_1,x,y_1} W_1(y_1|x)\left(\sum_{j_2}d(i_1j_2|y_1y_2)\right)\sum_{i_2} e(x|i_1i_2)\\
    &+ \frac{1}{2k_1k_2} \sum_{i_2,x,y_2} W_2(y_2|x)\left(\sum_{j_1}d(j_1i_2|y_1y_2)\right)\sum_{i_1} e(x|i_1i_2) \ .
  \end{aligned}
\end{equation}

Thus, by choosing $d_1(j_1|y_1) := \sum_{j_2}d(j_1j_2|y_1y_2)$ and $d_2(j_2|y_2) := \sum_{j_1}d(j_1j_2|y_1y_2)$, which are well-defined since $d$ is a non-signaling box, we have $\mathrm{S}_{\text{sum}}(W,k_1,k_2,e,d_1,d_2) = \mathrm{S}_{\text{sum}}^{\mathrm{NS}_{\text{dec}}}(W,k_1,k_2,e,d)$. By optimizing over all $e$ and $d$, $\mathrm{S}_{\text{sum}}^{\mathrm{NS}_{\text{dec}}}(W,k_1,k_2) \leq \mathrm{S}_{\text{sum}}(W,k_1,k_2)$. Since the inequality is obvious in the other direction, as $d(j_1j_2|y_1y_2) := d_1(j_1|y_1)d_2(j_2|y_2)$ is always a non-signaling box, we have that $\mathrm{S}_{\text{sum}}(W,k_1,k_2)=\mathrm{S}_{\text{sum}}^{\mathrm{NS}_{\text{dec}}}(W,k_1,k_2)$. This implies in particular that the capacity regions are the same, i.e. $\mathcal{C}_{\text{sum}}(W)=\mathcal{C}_{\text{sum}}^{\mathrm{NS}_{\text{dec}}}(W)$

Finally, since $\mathcal{C}(W) = \mathcal{C}_{\text{sum}}(W)$ by Proposition~\ref{prop:capacitysum} and $\mathcal{C}_{\text{sum}}^{\mathrm{NS}_{\text{dec}}}(W)= \mathcal{C}^{\mathrm{NS}_{\text{dec}}}(W)$ by Proposition~\ref{prop:NSdecoderssum}, we get that $\mathcal{C}(W) = \mathcal{C}^{\mathrm{NS}_{\text{dec}}}(W)$.
\end{proof}

\subsection{Full Non-Signaling Assistance}
In this section, we will consider the case where the sender and the receivers are given non-signaling assistance. This means that a three-party non-signaling box $P(xj_1j_2|(i_1i_2)y_1y_2)$ will replace the product $e(x|i_1i_2)d_1(j_1|y_1)d_2(j_2|y_2)$ in the previous objective values. A joint conditional probability $P(xj_1j_2|(i_1i_2)y_1y_2)$ is a non-signaling box if the marginal from any two parties is independent of the remaining party's input:
\begin{equation}
  \begin{aligned}
    &&\forall j_1,j_2,i_1,i_2,y_1,y_2,i_1',i_2', &&&\sum_{x} P(xj_1j_2|(i_1i_2)y_1y_2) = \sum_{x} P(xj_1j_2|(i_1'i_2')y_1y_2) \ ,\\
    &&\forall x,j_2,i_1,i_2,y_1,y_2,y_1', &&&\sum_{j_1} P(xj_1j_2|(i_1i_2)y_1y_2) = \sum_{j_1} P(xj_1j_2|(i_1i_2)y_1'y_2) \ ,\\
    &&\forall x,j_1,i_1,i_2,y_1,y_2,y_2', &&&\sum_{j_2} P(xj_1j_2|(i_1i_2)y_1y_2) = \sum_{j_2} P(xj_1j_2|(i_1i_2)y_1y_2') \ .
  \end{aligned}
\end{equation}

The scenario is depicted in Figure~\ref{fig:BCNScoding}.

\begin{figure}[!h]
\begin{center}
  \begin{tikzpicture}[auto, node distance=2cm,>=latex']
    \node [input, name=i1] {};
    \node [input, name=i2] {};
    \node [Bigblock, below of=i2] (P) {$\ \ \ e\ \ \ \ \ \ d_1\ \ \ \ \ \ d_2\ \ \ $};

    \draw (P.120) -- (P.240);
    \draw (P.60) -- (P.300);

    \draw [->] (P.90) -- node {$j_1$} +(0pt,1cm);
    \draw [->] (P.50) -- node {$j_2$} +(0pt,1cm);
    \coordinate (xbis) at ($ (P.130) + (0pt,1cm) $);
    \draw [-] (P.130) -- (xbis);
    \coordinate (y1) at ($ (P.270)+(0pt,-1.5cm) $);
    \draw [<-] (P.270) -- (y1);
    \coordinate (y2) at ($ (P.310)+(0pt,-2cm) $);
    \draw [<-] (P.310) -- (y2);
    \draw [<-] (P.230) -- node (S) {} +(0pt,-1cm);
    \coordinate (T) at ($ (S)+(1.2cm, 0pt) $);
    \node[left of=T] {$(i_1,i_2)$};

    \node [left of=P] (A) {};
    \node [bigblock, left of=A] (W) {$W$};
    \coordinate (y1bis) at ($ (y1)+(-4.38cm,0pt) $);
    \coordinate (y2bis) at ($ (y2)+(-4.4cm,0pt) $);
    \draw [-] (y1) -- node {$y_1$} (y1bis);
    \draw [-] (y2) -- node {$y_2$} (y2bis);
    \draw [-] (y1bis) -- (W.234);
    \draw [-] (y2bis) -- (W.303);
    \coordinate (x) at ($ (W.north)+(0pt,1.35cm) $) ;
    \draw[<-] (W.north) -- (x);
    \draw (x) -- node {$x$} (xbis);    
  \end{tikzpicture}
  \ \ \ \ 
  \begin{tikzpicture}[auto, node distance=2cm,>=latex']
    \node [input, name=i1] {};
    \node [input, name=i2] {};
    \node [Bigblock, below of=i2] (P) {$P(xj_1j_2|(i_1i_2)y_1y_2)$};

    \draw[dotted] (P.120) -- (P.240);
    \draw[dotted] (P.60) -- (P.300);

    \draw [->] (P.90) -- node {$j_1$} +(0pt,1cm);
    \draw [->] (P.50) -- node {$j_2$} +(0pt,1cm);
    \coordinate (xbis) at ($ (P.130) + (0pt,1cm) $);
    \draw [-] (P.130) -- (xbis);
    \coordinate (y1) at ($ (P.270)+(0pt,-1.5cm) $);
    \draw [<-] (P.270) -- (y1);
    \coordinate (y2) at ($ (P.310)+(0pt,-2cm) $);
    \draw [<-] (P.310) -- (y2);
    \draw [<-] (P.230) -- node (S) {} +(0pt,-1cm);
    \coordinate (T) at ($ (S)+(1.2cm, 0pt) $);
    \node[left of=T] {$(i_1,i_2)$};

    \node [left of=P] (A) {};
    \node [bigblock, left of=A] (W) {$W$};
    \coordinate (y1bis) at ($ (y1)+(-4.38cm,0pt) $);
    \coordinate (y2bis) at ($ (y2)+(-4.4cm,0pt) $);
    \draw [-] (y1) -- node {$y_1$} (y1bis);
    \draw [-] (y2) -- node {$y_2$} (y2bis);
    \draw [-] (y1bis) -- (W.234);
    \draw [-] (y2bis) -- (W.303);
    \coordinate (x) at ($ (W.north)+(0pt,1.35cm) $) ;
    \draw[<-] (W.north) -- (x);
    \draw (x) -- node {$x$} (xbis);    
  \end{tikzpicture}
\end{center}
\caption{A non-signaling box $P$ replacing $e,d_1,d_2$ in the coding problem for the broadcast channel $W$.}
\label{fig:BCNScoding}
\end{figure}
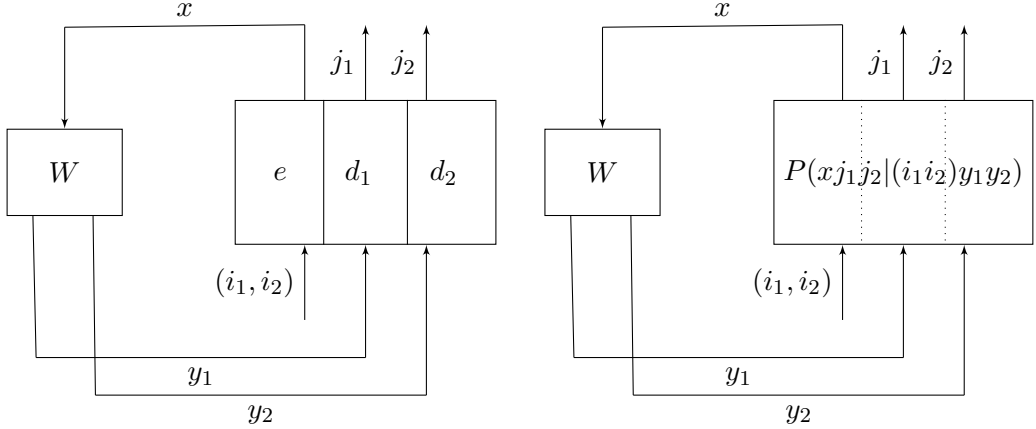

The cyclicity of Figure~\ref{fig:BCNScoding} is at first sight counter-intuitive. Note first that $P$ being a non-signaling box is completely independent of $W$: in particular, the variables $y_1,y_2$ do not need to follow any laws in the definition of $P$ being a non-signaling box. Therefore, the remaining ambiguity is the apparent need to encode and decode at the same time. However, since $P$ is a non-signaling box, we won't need to do both at the same time, although the global correlation between the sender and the receivers will be characterized only by $P(xj_1j_2|(i_1i_2)y_1y_2)$; see \cite{FF23} for a detailed discussion on that matter, the same paradox occurring for multiple-access channels and broadcast channels.


We will call the maximum sum success probability $\mathrm{S}_{\text{sum}}^{\mathrm{NS}}(W,k_1,k_2)$, which is given by the following linear program, where the constraints translate precisely the fact that $P$ is a non-signaling box:
\begin{equation}
  \begin{aligned}
    \mathrm{S}_{\text{sum}}^{\mathrm{NS}}(W,k_1,k_2) := &&\underset{P}{\maxi} &&& \frac{1}{2k_1k_2}\sum_{i_1,x,y_1} W_1(y_1|x)\sum_{i_2,j_2} P(xi_1j_2|(i_1i_2)y_1y_2)\\
    &&+&&& \frac{1}{2k_1k_2}\sum_{i_2,x,y_2} W_2(y_2|x) \sum_{i_1,j_1} P(xj_1i_2|(i_1i_2)y_1y_2)\\
    &&\st &&& \sum_{x} P(xj_1j_2|(i_1i_2)y_1y_2) = \sum_{x} P(xj_1j_2|(i_1'i_2')y_1y_2)\\
    &&&&& \sum_{j_1} P(xj_1j_2|(i_1i_2)y_1y_2) = \sum_{j_1} P(xj_1j_2|(i_1i_2)y_1'y_2)\\
    &&&&& \sum_{j_2} P(xj_1j_2|(i_1i_2)y_1y_2) = \sum_{j_2} P(xj_1j_2|(i_1i_2)y_1y_2')\\
    &&&&& \sum_{x,j_1,j_2} P(xj_1j_2|(i_1i_2)y_1y_2) = 1\\
    &&&&& P(xj_1j_2|(i_1i_2)y_1y_2) \geq 0
  \end{aligned}
\end{equation}

Since it is given as a linear program, the complexity of computing $\mathrm{S}_{\text{sum}}^{\mathrm{NS}}(W,k_1,k_2)$ is polynomial in the number of variables and constraints (see for instance Section 7.1 of~\cite{GM07}), which is a polynomial in $|\mathcal{X}|,|\mathcal{Y}_1|,|\mathcal{Y}_2|, k_1$ and $k_2$.

Similarly, we define the maximum joint success probability $\mathrm{S}^{\mathrm{NS}}(W,k_1,k_2)$ in the following way:
\begin{equation}
  \begin{aligned}
    \mathrm{S}^{\mathrm{NS}}(W,k_1,k_2) := &&\underset{P}{\maxi} &&& \frac{1}{k_1k_2}\sum_{i_1,i_2,x,y_1,y_2} W(y_1y_2|x)P(xi_1i_2|(i_1i_2)y_1y_2)\\
    &&\st &&& \sum_{x} P(xj_1j_2|(i_1i_2)y_1y_2) = \sum_{x} P(xj_1j_2|(i_1'i_2')y_1y_2)\\
    &&&&& \sum_{j_1} P(xj_1j_2|(i_1i_2)y_1y_2) = \sum_{j_1} P(xj_1j_2|(i_1i_2)y_1'y_2)\\
    &&&&& \sum_{j_2} P(xj_1j_2|(i_1i_2)y_1y_2) = \sum_{j_2} P(xj_1j_2|(i_1i_2)y_1y_2')\\
    &&&&& \sum_{x,j_1,j_2} P(xj_1j_2|(i_1i_2)y_1y_2) = 1\\
    &&&&& P(xj_1j_2|(i_1i_2)y_1y_2) \geq 0
  \end{aligned}
\end{equation}

We can rewrite both these programs in more convenient and smaller linear programs:

\begin{proposition}
\begin{equation}
  \begin{aligned}
    \mathrm{S}_{\text{sum}}^{\mathrm{NS}}(W,k_1,k_2) = &&\underset{p,r,r^1,r^2}{\maxi} &&& \frac{1}{2k_1k_2}\left(\sum_{x,y_1} W_1(y_1|x)r^1_{x,y_1} + \sum_{x,y_2} W_2(y_2|x)r^2_{x,y_2}\right)\\
    &&\st &&& \sum_{x} r_{x,y_1,y_2} = 1\\
    &&&&&\sum_{x} r^1_{x,y_1} = k_2\\
    &&&&& \sum_{x} r^2_{x,y_2} = k_1\\
    &&&&& \sum_{x} p_x = k_1k_2\\
    &&&&& 0 \leq r_{x,y_1,y_2} \leq r^1_{x,y_1}, r^2_{x,y_2} \leq p_x\\
    &&&&& p_x - r^1_{x,y_1} - r^2_{x,y_2} + r_{x,y_1,y_2} \geq 0\\
  \end{aligned}
\end{equation}

\begin{equation}
  \begin{aligned}
    \mathrm{S}^{\mathrm{NS}}(W,k_1,k_2) = &&\underset{p,r,r^1,r^2}{\maxi} &&& \frac{1}{k_1k_2}\sum_{x,y_1,y_2} W(y_1y_2|x)r_{x,y_1,y_2}\\
    &&\st &&& \sum_{x} r_{x,y_1,y_2} = 1\\
    &&&&&\sum_{x} r^1_{x,y_1} = k_2\\
    &&&&& \sum_{x} r^2_{x,y_2} = k_1\\
    &&&&& \sum_{x} p_x = k_1k_2\\
    &&&&& 0 \leq r_{x,y_1,y_2} \leq r^1_{x,y_1}, r^2_{x,y_2} \leq p_x\\
    &&&&& p_x - r^1_{x,y_1} - r^2_{x,y_2} + r_{x,y_1,y_2} \geq 0\\
  \end{aligned}
\end{equation}
\end{proposition}
\begin{proof}
  One can check that given a solution of the original program, the following choice of variables is a valid solution of the second program achieving the same objective value:
\begin{equation}
  \begin{aligned}
    &r_{x,y_1,y_2} := \sum_{i_1,i_2} P(xi_1i_2|(i_1i_2)y_1y_2)\ ,\\
    &r^1_{x,y_1} := \sum_{j_2,i_1,i_2} P(xi_1j_2|(i_1i_2)y_1y_2)\ ,\\
    &r^2_{x,y_2} := \sum_{j_1,i_1,i_2} P(xj_1i_2|(i_1i_2)y_1y_2)\ ,\\
    &p_x := \sum_{j_1,j_2,i_1,i_2} P(xj_1j_2|(i_1i_2)y_1y_2) \ .\\
  \end{aligned}
\end{equation}

For the other direction, given those variables, a non-signaling probability distribution $P(xj_1j_2|(i_1i_2)y_1y_2)$ is given by, for $j_1 \not= i_1$ and $j_2 \not= i_2$:
\begin{equation}
  \begin{aligned}
    &P(xi_1i_2|(i_1i_2)y_1y_2) = \frac{r_{x,y_1,y_2}}{k_1k_2}  \ ,\\
    &P(xj_1i_2|(i_1i_2)y_1y_2) = \frac{r^2_{x,y_2} - r_{x,y_1,y_2}}{k_1k_2(k_1-1)}  \ ,\\
    &P(xi_1j_2|(i_1i_2)y_1y_2) = \frac{r^1_{x,y_1} - r_{x,y_1,y_2}}{k_1k_2(k_2-1)} \ ,\\
    &P(xj_1j_2|(i_1i_2)y_1y_2) = \frac{p_{x} -  r^1_{x,y_1} - r^2_{x,y_2} + r_{x,y_1,y_2}}{k_1k_2(k_1-1)(k_2-1)} \ .\\
  \end{aligned}
\end{equation}
\end{proof}

As before, we define the (resp. sum) capacity region with non-signaling assistance using Definition~\ref{defi:generalCapacityRegion} by $\mathcal{C}^{\mathrm{NS}}(W) :=  \mathcal{C}[\mathrm{S}^{\mathrm{NS}}](W)$ (resp. $\mathcal{C}^{\mathrm{NS}}_{\text{sum}}(W) := \mathcal{C}[\mathrm{S}^{\mathrm{NS}}_{\text{sum}}](W)$).


\begin{proposition}
    For any broadcast channel $W$, $\mathcal{C}^{\mathrm{NS}}(W) = \mathcal{C}^{\mathrm{NS}}_{\text{sum}}(W)$.
\end{proposition}
\begin{proof}
Let us show that:
\[ 2 \mathrm{S}_{\text{sum}}^{\mathrm{NS}}(W,k_1,k_2)-1 \leq  \mathrm{S}^{\mathrm{NS}}(W,k_1,k_2) \leq \mathrm{S}_{\text{sum}}^{\mathrm{NS}}(W,k_1,k_2) \ . \]
This will imply in particular that:
\[ \underset{n \rightarrow +\infty}{\lim} \mathrm{S}^{\mathrm{NS}}(W^{\otimes n},\ceil{2^{R_1n}},\ceil{2^{R_2n}}) = 1 \iff \underset{n \rightarrow +\infty}{\lim} \mathrm{S}^{\mathrm{NS}}_{\text{sum}}(W^{\otimes n},\ceil{2^{R_1n}},\ceil{2^{R_2n}}) = 1 \ ,\]
thus define the same capacity region.

Let us consider an optimal solution $p_x,r_{x,y_1,y_2},r^1_{x,y_1},r^2_{x,y_2}$ of the program computing $\mathrm{S}_{\text{sum}}^{\mathrm{NS}}(W,k_1,k_2)$. We have:
\[ \mathrm{S}_{\text{sum}}^{\mathrm{NS}}(W,k_1,k_2) = \frac{1}{k_1k_2}\left(\sum_{x,y_1,y_2} W(y_1y_2|x)\frac{r^1_{x,y_1} + r^2_{x,y_2}}{2}\right) \ . \]

  However $r^1_{x,y_1} + r^2_{x,y_2} \leq p_x + r_{x,y_1,y_2}$ so we get that:
  \begin{equation}
    \begin{aligned}
      \mathrm{S}_{\text{sum}}^{\mathrm{NS}}(W,k_1,k_2) &\leq \frac{1}{2k_1k_2}\left(\sum_{x,y_1,y_2} W(y_1y_2|x)\left( p_x + r_{x,y_1,y_2}\right)\right)\\
      &= \frac{1}{2} + \frac{1}{2}\left[\frac{1}{k_1k_2}\left(\sum_{x,y_1,y_2} W(y_1y_2|x)r_{x,y_1,y_2}\right)\right]\\
      &\leq \frac{1}{2} + \frac{1}{2}\mathrm{S}^{\mathrm{NS}}(W,k_1,k_2) \ ,
    \end{aligned}
  \end{equation}
  since $p_x,r_{x,y_1,y_2},r^1_{x,y_1},r^2_{x,y_2}$ is a valid solution of the program computing $\mathrm{S}^{\mathrm{NS}}(W,k_1,k_2)$.

  On the other hand, consider now $p_x,r_{x,y_1,y_2},r^1_{x,y_1},r^2_{x,y_2}$ an optimal solution of the program computing $\mathrm{S}^{\mathrm{NS}}(W,k_1,k_2)$. We have that $r_{x,y_1,y_2} \leq r^1_{x,y_1},r^2_{x,y_2}$ so we have that $r_{x,y_1,y_2} \leq \frac{r^1_{x,y_1}+r^2_{x,y_2}}{2}$ and thus:
  \begin{equation}
    \begin{aligned}
      \mathrm{S}^{\mathrm{NS}}(W,k_1,k_2) &= \frac{1}{k_1k_2}\left(\sum_{x,y_1,y_2} W(y_1y_2|x)r_{x,y_1,y_2}\right)\\
      &\leq \frac{1}{k_1k_2}\left(\sum_{x,y_1,y_2} W(y_1y_2|x)\frac{r^1_{x,y_1}+r^2_{x,y_2}}{2}\right)\\
      &\leq \mathrm{S}_{\text{sum}}^{\mathrm{NS}}(W,k_1,k_2) \ ,
    \end{aligned}
  \end{equation}
  since $p_x,r_{x,y_1,y_2},r^1_{x,y_1},r^2_{x,y_2}$ is a valid solution of the program computing $\mathrm{S}^{\mathrm{NS}}_{\text{sum}}(W,k_1,k_2)$. This prove the inequalities $ 2 \mathrm{S}_{\text{sum}}^{\mathrm{NS}}(W,k_1,k_2)-1 \leq  \mathrm{S}^{\mathrm{NS}}(W,k_1,k_2) \leq \mathrm{S}_{\text{sum}}^{\mathrm{NS}}(W,k_1,k_2)$, and thus concludes the proof.
\end{proof}

\section{Approximation Algorithm for Deterministic Broadcast Channel Coding}
\label{section:ApproxDetBC}
In this section, we will address the question of the approximability of $\mathrm{S}(W,k_1,k_2)$, in the restricted scenario of a deterministic broadcast channel $W$. Specifically, we study the problem of finding a code  $e : [k_1] \times [k_2] \rightarrow \mathcal{X}$, $d_1 : \mathcal{Y}_1 \rightarrow [k_1]$,  $d_2 : \mathcal{Y}_2 \rightarrow [k_2]$ that maximizes the program computing $\mathrm{S}(W,k_1,k_2)$. Note that the restriction to deterministic codes does not affect the value of the objective of the program which is convex, and that the problem is as hard as finding any code maximizing the program computing $\mathrm{S}(W,k_1,k_2)$, as a deterministic code with a better or equal value can be retrieved easily from any code.

We say that $W$ is deterministic if $\forall x,y_1,y_2, W(y_1y_2|x) \in \{0,1\}$. We can then define $(W_1(x),W_2(x))$ as the only pair $(y_1,y_2)$ such that $W(y_1y_2|x)=1$, which exists uniquely as $W$ is a conditional probability distribution. Thus, the deterministic broadcast channel coding problem can be defined in the following way:

\begin{definition}[\textsc{DetBCC}]
  Given a deterministic channel $W$ and integers $k_1$ and $k_2$, the deterministic broadcast channel coding problem, which we call \textsc{DetBCC}, entails maximizing
  \[ \mathrm{S}(W,k_1,k_2,e,d_1,d_2) := \frac{1}{k_1k_2}\sum_{i_1,i_2}\mathbbm{1}_{d_1(W_1(e(i_1i_2)))=i_1}\mathbbm{1}_{d_2(W_2(e(i_1i_2)))=i_2}\]
  over all functions $e : [k_1] \times [k_2] \rightarrow \mathcal{X}$, $d_1 : \mathcal{Y}_1 \rightarrow [k_1]$,  $d_2 : \mathcal{Y}_2 \rightarrow [k_2]$.
\end{definition}

\subsection{Reformulation as a Bipartite Graph Problem}
We will reformulate \textsc{DetBCC} as a bipartite graph problem. But first, let us introduce some notations:

\begin{definition}[Graph notations]
  Consider a bipartite graph $G = (V_1 \sqcup V_2, E \subseteq V_1 \times V_2)$:
  \begin{enumerate}
    \item $G^{\mathcal{P}_1,\mathcal{P}_2}$, the quotient of $G$ by partitions $\mathcal{P}_1,\mathcal{P}_2$ of respectively $V_1,V_2$, is defined by:
      \[ G^{\mathcal{P}_1,\mathcal{P}_2} := \left(\mathcal{P}_1 \sqcup \mathcal{P}_2, \{(p_1,p_2) \in \mathcal{P}_1 \times \mathcal{P}_2: E \cap (p_1 \times p_2) \not= \emptyset\}\right) \ .\]
    \item $e_G(\mathcal{P}_1,\mathcal{P}_2) := |E^{G^{\mathcal{P}_1,\mathcal{P}_2}}|$ is the number of edges of $G^{\mathcal{P}_1,\mathcal{P}_2}$.
    \item $N_G^{\mathcal{P}_1,\mathcal{P}_2}(p) := N_{G^{\mathcal{P}_1,\mathcal{P}_2}}(p)$ is the set of neighbors of $p \in \mathcal{P}_1 \sqcup \mathcal{P}_2$ in the graph $G^{\mathcal{P}_1,\mathcal{P}_2}$.
    \item Similarly, $\deg_G^{\mathcal{P}_1,\mathcal{P}_2}(p) := \deg_{G^{\mathcal{P}_1,\mathcal{P}_2}}(p)$ is the degree, i.e. the number of neighbors, of $p$ in the graph $G^{\mathcal{P}_1,\mathcal{P}_2}$.
    \item We will use $V_1,V_2$ in previous notations when we do not partition on the left and right part respectively (or identify those to trivial partitions in singletons). For instance, $G^{V_1,V_2}=G$.
    \item We will use the notations $e(\mathcal{P}_1,\mathcal{P}_2)$, $N_{\mathcal{P}_1,\mathcal{P}_2}(p)$  and $\deg_{\mathcal{P}_1,\mathcal{P}_2}(p)$ when the graph $G$ considered is clear from context.
  \end{enumerate}
\end{definition}

Now, let us remark that a deterministic channel $W$, up to a permutation of elements of $\mathcal{X}$, is characterized by the following bipartite graph:

\begin{definition}[Bipartite Graph $G_W$ associated with the deterministic channel $W$]
  \label{defi:graphGW}
\[ G_W:=(\mathcal{Y}_1 \sqcup \mathcal{Y}_2, E = \{(y_1,y_2) \in \mathcal{Y}_1 \times \mathcal{Y}_2: \exists x \in \mathcal{X}, y_1=W_1(x) \text{ and } y_2=W_2(x)\}) \ . \]
\end{definition}

Indeed, permuting the elements of $\mathcal{X}$ does not change $G_W$ nor $\mathrm{S}(W,k_1,k_2)$. As a consequence, up to a multiplicative factor $k_1k_2$, we will show that \textsc{DetBCC} is equivalent to the following bipartite graph problem:

\begin{definition}[\textsc{DensestQuotientGraph}]
Given a bipartite graph $G  = (V_1 \sqcup V_2, E)$ and integers $k_1,k_2$, the problem \textsc{DensestQuotientGraph} entails maximizing $e_G(\mathcal{P}_1,\mathcal{P}_2)$, the number of edges of the quotient graph of $G$ by $\mathcal{P}_1,\mathcal{P}_2$, over all partitions $\mathcal{P}_1$ of $V_1$ in $k_1$ parts and $\mathcal{P}_2$ of $V_2$ in $k_2$ parts.
\end{definition}

\begin{proposition}
  \label{prop:BCCisDQG}
  Given a deterministic channel $W$ and integers $k_1,k_2$, it is equivalent to solve \textsc{DetBCC} on $W,k_1,k_2$ or \textsc{DensestQuotientGraph} on $G_W,k_1,k_2$. That is to say, given an optimal solution of one of those problems, one can efficiently construct an optimal solution of the other. Furthermore, their optimal values satisfy $k_1k_2\textsc{DetBCC}(W,k_1,k_2) = \textsc{DensestQuotientGraph}(G_W,k_1,k_2)$.
\end{proposition}
\begin{proof}
 Consider an optimal solution $e,d_1,d_2$ of \textsc{DetBCC}. Note that $d_1$ defines a partition $\mathcal{P}_1$ of $\mathcal{Y}_1$ in $k_1$ parts and $d_2$ defines a partition $\mathcal{P}_2$ of $\mathcal{Y}_2$ in $k_2$ parts, with $\mathcal{P}_b^{i_b} := \{y_b \in \mathcal{Y}_b : d_b(y) = i_b \}$ for $b \in \{1,2\}$. Then we have:
\begin{equation}
  \begin{aligned}
    k_1k_2\mathrm{S}(W,k_1,k_2,e,d_1,d_2) &= \sum_{i_1,i_2} \mathbbm{1}_{i_1=d_1(W_1(e(i_1,i_2)))}\mathbbm{1}_{i_2=d_2(W_2(e(i_1,i_2)))}\\
    &= \sum_{i_1,i_2} \mathbbm{1}_{W_1(e(i_1,i_2)) \in \mathcal{P}_1^{i_1}}\mathbbm{1}_{W_2(e(i_1,i_2)) \in \mathcal{P}_2^{i_2}} \ .
    \end{aligned}
\end{equation}

However, since we consider an optimal solution, we have that:
\[ \mathbbm{1}_{W_1(e(i_1,i_2)) \in \mathcal{P}_1^{i_1}}\mathbbm{1}_{W_2(e(i_1,i_2)) \in \mathcal{P}_2^{i_2}} = \max_{x \in \mathcal{X}}\mathbbm{1}_{W_1(x) \in \mathcal{P}_1^{i_1}}\mathbbm{1}_{W_2(x) \in \mathcal{P}_2^{i_2}} \ , \]
as $e(i_1,i_2)$ appears only here in the objective value. Thus:
\begin{equation}
  \begin{aligned}   
    k_1k_2\mathrm{S}(W,k_1,k_2,e,d_1,d_2) &=  \sum_{i_1,i_2} \max_{x \in \mathcal{X}}\mathbbm{1}_{W_1(x) \in \mathcal{P}_1^{i_1}}\mathbbm{1}_{W_2(x) \in \mathcal{P}_2^{i_2}}\\
    &= \sum_{i_1,i_2} \mathbbm{1}_{\exists (y_1,y_2) \in E^{G_W}: y_1 \in \mathcal{P}_1^{i_1}\text{ and } y_2 \in \mathcal{P}_2^{i_2}}\\
    &= \sum_{i_1,i_2} \mathbbm{1}_{E^{G_W} \cap \left(\mathcal{P}_1^{i_1} \times \mathcal{P}_2^{i_2}\right) \not= \emptyset}\\
    &=  e_{G_W}(\mathcal{P}_1,\mathcal{P}_2) \ ,\\
    \end{aligned}
\end{equation}
which proves that given an optimal solution of \textsc{DetBCC}, one can efficiently construct a solution $\mathcal{P}_1,\mathcal{P}_2$ of \textsc{DensestQuotientGraph} such that:
\[ e_{G_W}(\mathcal{P}_1,\mathcal{P}_2) =  k_1k_2\textsc{DetBCC}(W,k_1,k_2) \ . \]

For the other direction, consider an optimal solution $\mathcal{P}_1,\mathcal{P}_2$ of \textsc{DensestQuotientGraph}. We have as before that:
\[ e_{G_W}(\mathcal{P}_1,\mathcal{P}_2) =  \sum_{i_1,i_2} \max_{x \in \mathcal{X}}\mathbbm{1}_{W_1(x) \in \mathcal{P}_1^{i_1}}\mathbbm{1}_{W_2(x) \in \mathcal{P}_2^{i_2}} \ . \]

Now, let us define $e(i_1,i_2) \in \text{argmax}_{x \in \mathcal{X}}\mathbbm{1}_{W_1(x) \in \mathcal{P}_1^{i_1}}\mathbbm{1}_{W_2(x) \in \mathcal{P}_2^{i_2}}$ and $d_b(y_b)$ the index $i_b$ such that $y_b \in \mathcal{P}_b^{i_b}$, for $b \in \{1,2\}$. With those definitions, we get again that:
\begin{equation}
  \begin{aligned}
    \max_{x \in \mathcal{X}}\mathbbm{1}_{W_1(x) \in \mathcal{P}_1^{i_1}}\mathbbm{1}_{W_2(x) \in \mathcal{P}_2^{i_2}} &= \mathbbm{1}_{W_1(e(i_1,i_2)) \in \mathcal{P}_1^{i_1}}\mathbbm{1}_{W_2(e(i_1,i_2)) \in \mathcal{P}_2^{i_2}}\\
    &= \mathbbm{1}_{i_1=d_1(W_1(e(i_1,i_2)))}\mathbbm{1}_{i_2=d_2(W_2(e(i_1,i_2)))} \ ,
    \end{aligned}
\end{equation}
and thus we have:
\begin{equation}
  \begin{aligned}
e_{G_W}(\mathcal{P}_1,\mathcal{P}_2) &=  \sum_{i_1,i_2} \max_{x \in \mathcal{X}}\mathbbm{1}_{W_1(x) \in \mathcal{P}_1^{i_1}}\mathbbm{1}_{W_2(x) \in \mathcal{P}_2^{i_2}}\\
&= \sum_{i_1,i_2} \mathbbm{1}_{i_1=d_1(W_1(e(i_1,i_2)))}\mathbbm{1}_{i_2=d_2(W_2(e(i_1,i_2)))}\\
&= k_1k_2\mathrm{S}(W,k_1,k_2,e,d_1,d_2) \ ,
    \end{aligned}
\end{equation}
which proves that given an optimal solution of \textsc{DensestQuotientGraph}, one can efficiently construct a solution $e,d_1,d_2$ of \textsc{DetBCC} such that:
\[ k_1k_2\mathrm{S}(W,k_1,k_2,e,d_1,d_2) = \textsc{DensestQuotientGraph}(G_W,k_1,k_2) \ . \]

In particular, this implies that the optimal objective values satisfy:
\[ k_1k_2\textsc{DetBCC}(W,k_1,k_2) = \textsc{DensestQuotientGraph}(G_W,k_1,k_2) \ .\]

Therefore, the solutions of both problems constructed throughout the proof are in fact optimal.
\end{proof}

\begin{rk}
  Note that all bipartite graphs can be written as $G_W$ for some deterministic broadcast channel $W$, with $W$ unique up to a permutation of $\mathcal{X}$.
\end{rk}
\subsection{Approximation Algorithm for \textsc{DensestQuotientGraph}}
In this section, we will sort out how hard is \textsc{DensestQuotientGraph}, and thanks to Proposition~\ref{prop:BCCisDQG}, how hard is it to solve \textsc{DetBCC}.

\begin{theorem}
  \label{theo:DQGapprox}
  There exists a polynomial-time $(1-e^{-1})^2$-approximation algorithm for \textsc{DensestQuotientGraph}. Furthermore, it is \textrm{NP}-hard to solve exactly \textsc{DensestQuotientGraph}.
\end{theorem}

\begin{corollary}
  \label{theo:DetBCCapprox}
  There exists a polynomial-time $(1-e^{-1})^2$-approximation algorithm for \textsc{DetBCC}. Furthermore, it is \textrm{NP}-hard to solve exactly \textsc{DetBCC}.
\end{corollary}

The approximation algorithm is a two-step process. First, we consider the problem of maximizing $\sum_{i_2=1}^{k_2}\min\left(k_1,\deg_{V_1,\mathcal{P}_2}(\mathcal{P}_2^{i_2})\right)$ over all partitions $\mathcal{P}_2$ of $V_2$ in $k_2$ parts. We will show that this is a special case of the submodular welfare problem, which can be approximated within a factor $1-e^{-1}$ in polynomial time~\cite{Vondrak08}. We then choose the partition $\mathcal{P}_1$ on $V_1$ in $k_1$ parts uniformly at random. This partition pair will give an objective value $e(\mathcal{P}_1,\mathcal{P}_2)$ within a $(1-e^{-1})^2$ factor from the optimal solution in expectation.

\begin{proof}[Proof of Theorem~\ref{theo:DQGapprox}]
  Consider first the hardness result. Let us show that the decision version of \textsc{DensestQuotientGraph} is \textrm{NP}-complete. It is in \textrm{NP}, the certificate being the two partitions and the selection of edges between those partitions. It is \textrm{NP}-hard as one of its particular cases is the \textsc{SetSplitting} problem (see for instance~\cite{GJ79}), in the case where $k_1=2$ and $k_2=|V_2|$, by interpreting the neighbors of $v_2 \in V_2$ as a set covering elements of $V_1$.

We will show nonetheless that this problem can be approximated within a factor $(1-e^{-1})^2$ in polynomial time. First we consider the case where $k_2=|V_2|$. We can then always assume that the right partition is $\mathcal{P}_2 := \{ \{v_2\} : v_2 \in V_2 \}$, which leads necessarily to a greater or equal number of edges in the quotient graph that with any other right partition. So, in that setting, we only need to find a partition of $V_1$ in $k_1$ parts maximizing the number of edges between vertices in the right part and the quotient of the left vertices.

First, the maximum value we can get is upper bounded by $\sum_{v_2 \in V_2}\min\left(k_1,\deg(v_2)\right)$. Indeed, each vertex of $v_2$ can be connected at most to the $k_1$ parts of $V_1$, so its contribution is bounded by $k_1$, and there needs to be an edge to each part it is connected, so its contribution is also bounded by $\deg(v_2)$. Let us show that if we take a partition $\mathcal{P}_1$ of $V_1$ uniformly at random, we get:
\begin{equation}
  \begin{aligned}
  \mathbb{E}_{\mathcal{P}_1}[e(\mathcal{P}_1,V_2)] &\geq \left(1-\left(1-\frac{1}{k_1}\right)^{k_1}\right)\sum_{v_2 \in V_2}\min\left(k_1,\deg(v_2)\right)\\
  &\geq (1-e^{-1})\max_{\mathcal{P}_1}e(\mathcal{P}_1,V_2) \ .
  \end{aligned}
\end{equation}  

We have $e(\mathcal{P}_1,V_2)=\sum_{v_2 \in V_2}\deg_{\mathcal{P}_1,V_2}(v_2)$, so by linearity of expectation $\mathbb{E}_{\mathcal{P}_1}[e(\mathcal{P}_1,V_2)] = \sum_{v_2 \in V_2}\mathbb{E}_{\mathcal{P}_1}[\deg_{\mathcal{P}_1,V_2}(v_2)]$. However $\deg_{\mathcal{P}_1,V_2}(v_2)=|\{i_1 \in [k_1]: N(v_2) \cap \mathcal{P}_1^{i_1} \not= \emptyset  \}|$. Recall also that for any $v_1$, $\mathbb{P}\left(v_1 \in \mathcal{P}_1^{i_1}\right) = \frac{1}{k_1}$ since the partition is taken uniformly at random. Thus, we get:
\begin{equation}
  \begin{aligned}
    \mathbb{E}_{\mathcal{P}_1}[\deg_{\mathcal{P}_1,V_2}(v_2)] &= \mathbb{E}_{\mathcal{P}_1}\left[|\{i_1 \in [k_1]: N(v_2) \cap \mathcal{P}_1^{i_1} \not= \emptyset  \}|\right]
    = \mathbb{E}_{\mathcal{P}_1}\left[\sum_{i_1=1}^{k_1} \mathbbm{1}_{N(v_2) \cap \mathcal{P}_1^{i_1} \not= \emptyset}\right]\\
    &= \sum_{i_1=1}^{k_1} \mathbb{E}_{\mathcal{P}_1}\left[\mathbbm{1}_{N(v_2) \cap \mathcal{P}_1^{i_1} \not= \emptyset}\right]
    = \sum_{i_1=1}^{k_1}\mathbb{P}\left(N(v_2) \cap \mathcal{P}_1^{i_1} \not= \emptyset\right)\\
    &= \sum_{i_1=1}^{k_1}\left(1 - \mathbb{P}\left(N(v_2) \cap \mathcal{P}_1^{i_1} = \emptyset\right)\right)\\
    &= \sum_{i_1=1}^{k_1}\left(1 - \prod_{v_1 \in N(v_2)}\mathbb{P}\left(v_1 \not\in \mathcal{P}_1^{i_1}\right)\right)\\
    &= \sum_{i_1=1}^{k_1}\left(1 - \prod_{v_1 \in N(v_2)}\mathbb{P}\left(v_1 \not\in \mathcal{P}_1^{i_1}\right)\right) = k_1\left(1-\left(1-\frac{1}{k_1}\right)^{\deg(v_2)}\right) \ ,
  \end{aligned}
\end{equation}
since $\mathbb{P}\left(v_1 \not\in \mathcal{P}_1^{i_1}\right)=1-\frac{1}{k_1}$ and $|N(v_2)|=\deg(v_2)$. So, in all:
\[ \mathbb{E}_{\mathcal{P}_1}[e(\mathcal{P}_1,V_2)] = \sum_{v_2 \in V_2}\mathbb{E}_{\mathcal{P}_1}[\deg_{\mathcal{P}_1,V_2}(v_2)] = k_1\sum_{v_2 \in V_2}\left(1-\left(1-\frac{1}{k_1}\right)^{\deg(v_2)}\right) \ . \]

However, the function $f : x \mapsto 1-\left(1-\frac{1}{k_1}\right)^x$ is nondecreasing concave with $f(0)=0$, so $\frac{f(x)}{x} \geq \frac{f(y)}{y}$ for $x \leq y$. In particular, we have that:
\[ f(\min(k_1,\deg(v_2))) \geq \frac{\min(k_1,\deg(v_2))}{k_1}f(k_1) \ , \]
and thus:
\begin{equation}
  \begin{aligned}
    \mathbb{E}_{\mathcal{P}_1}[e(\mathcal{P}_1,V_2)] &\geq k_1\sum_{v_2 \in V_2}\left(1-\left(1-\frac{1}{k_1}\right)^{\min(k_1,\deg(v_2))}\right)\\
    &\geq k_1\frac{\sum_{v_2 \in V_2}\min(k_1,\deg(v_2))}{k_1}\left(1-\left(1-\frac{1}{k_1}\right)^{k_1}\right)\\
    &\geq \left(1-\left(1-\frac{1}{k_1}\right)^{k_1}\right)\sum_{v_2 \in V_2}\min\left(k_1,\deg(v_2)\right)\\
    &\geq (1-e^{-1})\max_{\mathcal{P}_1}e(\mathcal{P}_1,V_2) \ .
  \end{aligned}
\end{equation}

Let us now consider the general case with $k_2$ unconstrained. We apply the previous discussion on the graph $G^{V_1,\mathcal{P}_2}$ for some fixed partition $\mathcal{P}_2$ of $V_2$. Since $e_{G^{V_1,\mathcal{P}_2}}(\mathcal{P}_1,\mathcal{P}_2) = e(\mathcal{P}_1,\mathcal{P}_2)$, we have the upper bound:

\[ \max_{\mathcal{P}_1} e(\mathcal{P}_1,\mathcal{P}_2) \leq \sum_{i_2=1}^{k_2}\min\left(k_1,\deg_{V_1,\mathcal{P}_2}(\mathcal{P}_2^{i_2})\right) \ , \]
and the previous algorithm gives us a partition $\mathcal{P}_1$ of $V_1$ such that:
\[ e(\mathcal{P}_1,\mathcal{P}_2) \geq (1-e^{-1})\sum_{i_2=1}^{k_2}\min\left(k_1,\deg_{V_1,\mathcal{P}_2}(\mathcal{P}_2^{i_2})\right)\ . \]

Therefore, let us focus on the following optimization problem:
\[\max_{\mathcal{P}_2}\sum_{i_2=1}^{k_2}\min\left(k_1,\deg_{V_1,\mathcal{P}_2}(\mathcal{P}_2^{i_2})\right) \ , \]

We will give a $(1-e^{-1})$-approximation algorithm running in polynomial time for this problem. In all, this will allow us to get in polynomial time a partition pair $(\mathcal{P}_1,\mathcal{P}_2)$ such that:
\begin{equation}
  \begin{aligned}
    e(\mathcal{P}_1,\mathcal{P}_2) &\geq (1-e^{-1})\sum_{i_2=1}^{k_2}\min\left(k_1,\deg_{V_1,\mathcal{P}_2}(\mathcal{P}_2^{i_2})\right)\\
    &\geq (1-e^{-1})^2\max_{\mathcal{P}_2}\sum_{i_2=1}^{k_2}\min\left(k_1,\deg_{V_1,\mathcal{P}_2}(\mathcal{P}_2^{i_2})\right)\\
    &\geq (1-e^{-1})^2\max_{\mathcal{P}_1,\mathcal{P}_2} e(\mathcal{P}_1,\mathcal{P}_2) \ .
  \end{aligned}
\end{equation}

The problem $\max_{\mathcal{P}_2}\sum_{i_2=1}^{k_2}\min\left(k_1,\deg_{V_1,\mathcal{P}_2}(\mathcal{P}_2^{i_2})\right)$ is a particular instance of the submodular welfare problem from~\cite{Vondrak08}. Note that $\deg_{V_1,\mathcal{P}_2}(\mathcal{P}_2^{i_2}) = \deg_{V_1,\{\mathcal{P}_2^{i_2}, V_2-\mathcal{P}_2^{i_2}\}}(\mathcal{P}_2^{i_2})$, as the degree of $\mathcal{P}_2^{i_2}$ does not depend on the rest of the partition $\mathcal{P}_2$. Then, $h(S_2) := \min\left(k_1,\deg_{V_1,\{S_2,V_2-S_2\}}(S_2)\right)$, for $S_2 \subseteq V_2$, is a nondecreasing submodular function, as $S_2 \mapsto \deg_{V_1,\{S_2,V_2-S_2\}}(S_2)$ is a nondecreasing submodular function on $V_2$ and $\min(k_1,\cdot)$ is nondecreasing concave. Thus, we want to maximize $\sum_{i_2=1}^{k_2}h(S_{i_2})$ where $(S_{i_2})_{i_2 \in [k_2]}$ is a partition of items in $V_2$ among $k_2$ bidders. It is a particular case of the submodular welfare problem where each nondecreasing submodular utility weight is the same for all bidders and equal to $h$. Thus, thanks to~\cite{Vondrak08}, there exists a polynomial-time $(1-e^{-1})$-approximation of $\max_{\mathcal{P}_2}\sum_{i_2=1}^{k_2}\min\left(k_1,\deg_{V_1,\mathcal{P}_2}(\mathcal{P}_2^{i_2})\right)$.
\end{proof}

\subsection{Non-Signaling Assisted Capacity Region for Deterministic Channels}
Thanks to Theorem~\ref{theo:DQGapprox} and Proposition~\ref{prop:BCCisDQG}, there exists a constant-factor approximation algorithm for the broadcast channel coding problem running in polynomial time. We aim to show here that the non-signaling assisted value is linked by a constant factor to the unassisted one. Indeed, the hope is that the non-signaling assisted program is linked to the linear relaxation of the unassisted problem, thus is likely a good approximation since the broadcast channel coding problem can be approximated in polynomial time.

This turns out to be true, and will be proved through the following theorem:

\begin{theorem}
  \label{theo:NSdet}
  If $W$ is a deterministic broadcast channel, then for all $\ell_1 \leq k_1$ and $\ell_2 \leq k_2$:
  \[ \left(1 - \frac{k_1^{k_1}e^{-k_1}}{k_1!}\right)\left(1-\left(1-\frac{1}{\ell_1}\right)^{k_1}\right)\left(1-\left(1-\frac{1}{\ell_2}\right)^{k_2}\right)\mathrm{S}^{\mathrm{NS}}(W,k_1,k_2) \leq \mathrm{S}(W,\ell_1,\ell_2)\ . \]
\end{theorem}

\begin{corollary}
  \label{cor:NSdet}
  For any deterministic broadcast channel $W$, $\mathcal{C}^{\mathrm{NS}}(W)=\mathcal{C}(W)$.
\end{corollary}
\begin{proof}
  We apply Theorem~\ref{theo:NSdet} on the deterministic broadcast channel $W^{\otimes n}$.

  We fix $k_1=2^{nR_1},k_2=2^{nR_2}$ and $\ell_1=\frac{2^{nR_1}}{n},\ell_2=\frac{2^{nR_2}}{n}$. Since $1-\left(1-\frac{1}{\ell}\right)^{k} \geq 1-e^{-\frac{k}{\ell}}$, we get:
\[ \left(1 - \frac{k_1^{k_1}e^{-k_1}}{k_1!}\right)\left(1-e^{-n}\right)^2\mathrm{S}^{\mathrm{NS}}(W^{\otimes n},2^{nR_1},2^{nR_1}) \leq \mathrm{S}\left(W^{\otimes n},\frac{2^{nR_1}}{n},\frac{2^{nR_2}}{n}\right)\ . \]

  As $\left(1 - \frac{k_1^{k_1}e^{-k_1}}{k_1!}\right)\left(1-e^{-n}\right)^2$ tends to $1$ when $n$ tends to infinity, we get $\forall \varepsilon > 0$, $\exists N \in \mathbb{N}$, $\forall n\geq N$:
\[ (1-\varepsilon)\mathrm{S}^{\mathrm{NS}}(W^{\otimes n},2^{nR_1},2^{nR_1}) \leq \mathrm{S}(W^{\otimes n},2^{n(R_1-\frac{\log(n)}{n})},2^{n(R_2-\frac{\log(n)}{n})}) \ . \]

Thus, if $\underset{n \rightarrow +\infty}{\lim} \mathrm{S}^{\mathrm{NS}}(W^{\otimes n},2^{nR_1},2^{nR_1})  = 1$, we have that for all $R_1'<R_1$ and $R_2'<R_2$:
\[ \underset{n \rightarrow +\infty}{\lim} \mathrm{S}(W^{\otimes n},2^{nR_1'},2^{nR_1'}) \geq 1-\varepsilon \ . \]

Since this is true for all $\varepsilon > 0$, we get in fact that $\underset{n \rightarrow +\infty}{\lim} \mathrm{S}(W^{\otimes n},2^{nR_1'},2^{nR_1'}) = 1$. This implies that $\mathcal{C}^{\mathrm{NS}}(W) \subseteq \mathcal{C}(W)$, and thus that the capacity regions are equal as the other inclusion is always satisfied.

\end{proof}

Let us now prove the main result:

\begin{proof}[Proof of Theorem~\ref{theo:NSdet}]
  The proof will be done in three parts. We will work on the graph $G_W$ (see Definition~\ref{defi:graphGW}).
  \begin{enumerate}
  \item First, we prove that for any partition $\mathcal{P}_2$ of $\mathcal{Y}_2$ in $\ell_2$ parts:
    \[ \mathrm{S}(W,\ell_1,\ell_2) \geq \left(1-\left(1-\frac{1}{\ell_1}\right)^{k_1}\right)\frac{\sum_{i_2=1}^{\ell_2}\min\left(k_1,\deg_{\mathcal{Y}_1,\mathcal{P}_2}(\mathcal{P}_2^{i_2})\right)}{k_1\ell_2} \ . \]
  \item Then, we show that there exists a partition $\mathcal{P}_2$ such that:
    \begin{equation}
      \begin{aligned}
        &\frac{\sum_{i_2=1}^{\ell_2}\min\left(k_1,\deg_{\mathcal{Y}_1,\mathcal{P}_2}(\mathcal{P}_2^{i_2})\right)}{k_1\ell_2}\\
        &\geq \left(1 - \frac{k_1^{k_1}e^{-k_1}}{k_1!}\right)\left(1-\left(1-\frac{1}{\ell_2}\right)^{k_2}\right)\frac{\min\left(k_1k_2,\sum_{y_1}\min(k_2,\deg(y_1))\right)}{k_1k_2} \ .
      \end{aligned}
    \end{equation}
  \item Finally, we prove that:
    \[ \frac{\min\left(k_1k_2,\sum_{y_1}\min(k_2,\deg(y_1))\right)}{k_1k_2} \geq \mathrm{S}^{\textrm{NS}}(W,k_1,k_2) \ . \]
  \end{enumerate}

  By combining these three inequalities, we get precisely the claimed result.
  
  \begin{enumerate}
  \item This part shares a lot of similarities with the proof of Theorem~\ref{theo:DQGapprox}, which we will adapt to this particular situation. Let us show that if we take a partition $\mathcal{P}_1$ of $\mathcal{Y}_1$ of size $\ell_1$ uniformly at random, we get, for some fixed $\mathcal{P}_2$ of size $\ell_2$:
    \[ \mathbb{E}_{\mathcal{P}_1}[e_{G_W}(\mathcal{P}_1,\mathcal{P}_2)] \geq \frac{\ell_1}{k_1}\left(1-\left(1-\frac{1}{\ell_1}\right)^{k_1}\right)\sum_{i_2=1}^{\ell_2}\min\left(k_1,\deg_{\mathcal{Y}_1,\mathcal{P}_2}(\mathcal{P}_2^{i_2})\right) \ .\]

    Since $\ell_1\ell_2 \mathrm{S}(W,\ell_1,\ell_2) = \underset{\mathcal{P}_1 \text{ in $\ell_1$ parts},\mathcal{P}_2 \text{ in $\ell_2$ parts}}{\maxi} \ e_{G_W}(\mathcal{P}_1,\mathcal{P}_2)$ by Proposition~\ref{prop:BCCisDQG}, this will imply that:
    \begin{equation}
      \begin{aligned}
        \mathrm{S}(W,\ell_1,\ell_2) &\geq \frac{1}{\ell_1\ell_2}\mathbb{E}_{\mathcal{P}_1}[e_{G_W}(\mathcal{P}_1,\mathcal{P}_2)]\\
        &\geq \left(1-\left(1-\frac{1}{\ell_1}\right)^{k_1}\right)\frac{\sum_{i_2=1}^{\ell_2}\min\left(k_1,\deg_{\mathcal{Y}_1,\mathcal{P}_2}(\mathcal{P}_2^{i_2})\right)}{k_1\ell_2} \ .
      \end{aligned}
    \end{equation}

We have that $e_{G_W}(\mathcal{P}_1,\mathcal{P}_2) = \sum_{i_2=1}^{\ell_2}\deg_{\mathcal{P}_1,\mathcal{P}_2}(\mathcal{P}_2^{i_2})$, so by linearity of expectation, we have that $\mathbb{E}_{\mathcal{P}_1}[e_{G_W}(\mathcal{P}_1,\mathcal{P}_2)] = \sum_{i_2=1}^{\ell_2}\mathbb{E}_{\mathcal{P}_1}[\deg_{\mathcal{P}_1,\mathcal{P}_2}(\mathcal{P}_2^{i_2})]$, so we will focus on the contribution of one particular $\mathcal{P}_2^{i_2}$.

Then, we have that $\deg_{\mathcal{P}_1,\mathcal{P}_2}(\mathcal{P}_2^{i_2})=|\{i_1 \in [\ell_1]: N_{\mathcal{Y}_1,\mathcal{P}_2}(\mathcal{P}_2^{i_2}) \cap \mathcal{P}_1^{i_1} \not= \emptyset  \}|$. Recall that $\mathbb{P}\left(v_1 \in \mathcal{P}_1^{i_1}\right) = \frac{1}{\ell_1}$ for any $v_1$ since the partition is taken uniformly at random. Thus:
\begin{equation}
  \begin{aligned}
    &\mathbb{E}_{\mathcal{P}_1}[\deg_{\mathcal{P}_1,\mathcal{P}_2}(\mathcal{P}_2^{i_2})] = \mathbb{E}_{\mathcal{P}_1}\left[|\{i_1 \in [\ell_1]: N_{\mathcal{Y}_1,\mathcal{P}_2}(\mathcal{P}_2^{i_2}) \cap \mathcal{P}_1^{i_1} \not= \emptyset  \}|\right]\\
    &= \mathbb{E}_{\mathcal{P}_1}\left[\sum_{i_1=1}^{\ell_1} \mathbbm{1}_{N_{\mathcal{Y}_1,\mathcal{P}_2}(\mathcal{P}_2^{i_2}) \cap \mathcal{P}_1^{i_1} \not= \emptyset}\right]
    = \sum_{i_1 = 1}^{\ell_1} \mathbb{E}_{\mathcal{P}_1}\left[\mathbbm{1}_{N_{\mathcal{Y}_1,\mathcal{P}_2}(\mathcal{P}_2^{i_2}) \cap \mathcal{P}_1^{i_1} \not= \emptyset}\right]\\
    &= \sum_{i_1 = 1}^{\ell_1}\mathbb{P}\left(N_{\mathcal{Y}_1,\mathcal{P}_2}(\mathcal{P}_2^{i_2}) \cap \mathcal{P}_1^{i_1} \not= \emptyset\right)
    = \sum_{i_1 = 1}^{\ell_1}\left(1 - \mathbb{P}\left(N_{\mathcal{Y}_1,\mathcal{P}_2}(\mathcal{P}_2^{i_2}) \cap \mathcal{P}_1^{i_1} = \emptyset\right)\right)\\
    &= \sum_{i_1 = 1}^{\ell_1}\left(1 - \prod_{v_1 \in N(\mathcal{P}_2^{i_2})}\mathbb{P}\left(v_1 \not\in \mathcal{P}_1^{i_1})\right)\right) = \ell_1\left(1-\left(1-\frac{1}{\ell_1}\right)^{\deg_{\mathcal{Y}_1,\mathcal{P}_2}(\mathcal{P}_2^{i_2})}\right) \ .
  \end{aligned}
\end{equation}

So, in all we have that:
\begin{equation}
  \begin{aligned}
    \mathbb{E}_{\mathcal{P}_1}[e_{G_W}(\mathcal{P}_1,\mathcal{P}_2)] &= \sum_{i_2=1}^{\ell_2}\mathbb{E}_{\mathcal{P}_1}[\deg_{\mathcal{P}_1,\mathcal{P}_2}(\mathcal{P}_2^{i_2})]\\
    &= \ell_1\sum_{i_2=1}^{\ell_2}\left(1-\left(1-\frac{1}{\ell_1}\right)^{\deg_{\mathcal{Y}_1,\mathcal{P}_2}(\mathcal{P}_2^{i_2})}\right) \ .
  \end{aligned}
\end{equation}

However the function $f : x \mapsto 1-\left(1-\frac{1}{\ell_1}\right)^x$ is nondecreasing concave with $f(0)=0$, so $\frac{f(x)}{x} \geq \frac{f(y)}{y}$ for $x \leq y$. In particular, we have that:
\[ f(\min(k_1,\deg_{\mathcal{Y}_1,\mathcal{P}_2}(\mathcal{P}_2^{i_2}))) \geq \frac{\min(k_1,\deg_{\mathcal{Y}_1,\mathcal{P}_2}(\mathcal{P}_2^{i_2})))}{k_1}f(k_1) \ , \]
and thus:
\begin{equation}
  \begin{aligned}
    \mathbb{E}_{\mathcal{P}_1}[e_{G_W}(\mathcal{P}_1,\mathcal{P}_2)] &\geq \ell_1\sum_{i_2=1}^{\ell_2}\left(1-\left(1-\frac{1}{\ell_1}\right)^{\min(k_1,\deg_{\mathcal{Y}_1,\mathcal{P}_2}(\mathcal{P}_2^{i_2}))}\right)\\
    &\geq \ell_1\frac{\sum_{i_2=1}^{\ell_2}\min(k_1,\deg_{\mathcal{Y}_1,\mathcal{P}_2}(\mathcal{P}_2^{i_2}))}{k_1}\left(1-\left(1-\frac{1}{\ell_1}\right)^{k_1}\right)\\
    &= \frac{\ell_1}{k_1}\left(1-\left(1-\frac{1}{\ell_1}\right)^{k_1}\right)\sum_{i_2=1}^{\ell_2}\min\left(k_1,\deg_{\mathcal{Y}_1,\mathcal{P}_2}(\mathcal{P}_2^{i_2})\right) \ ,
  \end{aligned}
\end{equation} 
which concludes the first part of the proof.
  
\item Let us take $\mathcal{P}_2$ a partition of $\mathcal{Y}_2$ of size $\ell_2$ uniformly at random, and let us prove that:
  \[ \mathbb{E}\left[\sum_{i_2=1}^{\ell_2}\min\left(k_1,\deg_{\mathcal{Y}_1,\mathcal{P}_2}(\mathcal{P}_2^{i_2})\right)\right] \]
  is greater than or equal to:
  \[ \frac{\ell_2}{k_2}\left(1 - \frac{k_1^{k_1}e^{-k_1}}{k_1!}\right)\left(1-\left(1-\frac{1}{\ell_2}\right)^{k_2}\right)\min\left(k_1k_2,\sum_{y_1}\min(k_2,\deg(y_1))\right) \ . \]

  First, $\sum_{i_2=1}^{\ell_2}\min\left(k_1,\deg_{\mathcal{Y}_1,\mathcal{P}_2}(\mathcal{P}_2^{i_2})\right) = \sum_{i_2=1}^{\ell_2}\varphi(\deg_{\mathcal{Y}_1,\mathcal{P}_2}(\mathcal{P}_2^{i_2}))$ with $\varphi(j):=\min(k_1,j)$ which is a concave function. The Poisson concavity ratio, introduced in~\cite{BFF21}, is defined by $\alpha_{\varphi} = \inf_{x \in \mathbb{R}^+} \frac{\mathbb{E}[\varphi(\Poi(x))]}{\varphi(x)}$ and is equal to $1 - \frac{k_1^{k_1}e^{-k_1}}{k_1!}$ for that particular function~\cite{BFF21}. We will use the following property from~\cite{BFF21}:
  \begin{proposition}[Lemma 2.2 from~\cite{BFF21}]
    For $\varphi$ concave, and $p \in [0,1]^m$, we have:
    \[\mathbb{E}\left[\varphi\left(\sum_{i=1}^m\Ber(p_i)\right)\right] \geq \mathbb{E}\left[\varphi\left(\Poi\left(\sum_{i=1}^m p_i\right)\right)\right]\ .\]
  \label{prop:ConvexOrder}
  \end{proposition}
  
  Let us find the law of $\deg_{\mathcal{Y}_1,\mathcal{P}_2}(\mathcal{P}_2^{i_2})$:
  \begin{equation}
    \begin{aligned}
      \deg_{\mathcal{Y}_1,\mathcal{P}_2}(\mathcal{P}_2^{i_2}) &= \sum_{y_1} \mathbbm{1}_{N(y_1) \cap \mathcal{P}_2^{i_2} \not= \emptyset} = \sum_{y_1} \left(1-\mathbbm{1}_{N(y_1) \cap \mathcal{P}_2^{i_2} = \emptyset}\right)\\
      &= \sum_{y_1} \left(1-\mathbbm{1}_{\forall y_2 \in N(y_1), y_2 \not\in \mathcal{P}_2^{i_2}}\right) = \sum_{y_1} \Ber\left(1-\left(1-\frac{1}{\ell_2}\right)^{\deg(y_1)}\right)
    \end{aligned}
  \end{equation}

Thus:
\begin{equation}
  \begin{aligned}
    &\mathbb{E}\left[\varphi(\deg_{\mathcal{Y}_1,\mathcal{P}_2}(\mathcal{P}_2^{i_2}))\right] = \mathbb{E}\left[\varphi\left(\sum_{y_1}\Ber\left(1-\left(1-\frac{1}{\ell_2}\right)^{\deg(y_1)}\right)\right)\right]\\
    &\geq  \mathbb{E}\left[\varphi\left(\Poi\left(\sum_{y_1}\left(1-\left(1-\frac{1}{\ell_2}\right)^{\deg(y_1)}\right)\right)\right)\right] \text{ by Proposition~\ref{prop:ConvexOrder}}\\
    &\geq \alpha_{\varphi}\varphi\left(\sum_{y_1}\left(1-\left(1-\frac{1}{\ell_2}\right)^{\deg(y_1)}\right)\right) \text{ by definition of $\alpha_{\varphi}$.}
  \end{aligned}
\end{equation}

But:
\begin{equation}
  \begin{aligned}
    \sum_{y_1}\left(1-\left(1-\frac{1}{\ell_2}\right)^{\deg(y_1)}\right) &\geq \sum_{y_1}\left(1-\left(1-\frac{1}{\ell_2}\right)^{\min\left(k_2,\deg(y_1)\right)}\right)\\
    &\geq \left(1-\left(1-\frac{1}{\ell_2}\right)^{k_2}\right)\frac{1}{k_2}\sum_{y_1}\min\left(k_2,\deg(y_1)\right) \ ,
  \end{aligned}
\end{equation}
as before. Since $\varphi$ is concave and $\varphi(0)=0$, we have in particular that for all $0 \leq c \leq 1$ and $x \in \mathbb{R}, \varphi(cx)\geq c\varphi(x)$. We know also that $\varphi$ is nondecreasing. This implies that:
\begin{equation}
  \begin{aligned}
    &\varphi\left(\sum_{y_1}\left(1-\left(1-\frac{1}{\ell_2}\right)^{\deg(y_1)}\right)\right)\\
    &\geq \varphi\left(\left(1-\left(1-\frac{1}{\ell_2}\right)^{k_2}\right)\frac{1}{k_2}\sum_{y_1}\min\left(k_2,\deg(y_1)\right)\right) \\
&\geq \left(1-\left(1-\frac{1}{\ell_2}\right)^{k_2}\right)\varphi\left(\frac{1}{k_2}\sum_{y_1}\min\left(k_2,\deg(y_1)\right)\right) \ ,
  \end{aligned}
\end{equation}
as $0 \leq 1-\left(1-\frac{1}{\ell_2}\right)^{k_2} \leq 1$. Thus $\mathbb{E}\left[\varphi(\deg_{\mathcal{Y}_1,\mathcal{P}_2}(\mathcal{P}_2^{i_2})))\right]$ is larger than or equal to:
\begin{equation}
  \begin{aligned}
    &\alpha_{\varphi}\left(1-\left(1-\frac{1}{\ell_2}\right)^{k_2}\right)\min\left(k_1,\frac{1}{k_2}\sum_{y_1}\min\left(k_2,\deg(y_1)\right)\right)\\
    &= \frac{1}{k_2}\left(1 - \frac{k_1^{k_1}e^{-k_1}}{k_1!}\right)\left(1-\left(1-\frac{1}{\ell_2}\right)^{k_2}\right)\min\left(k_1k_2,\sum_{y_1}\min\left(k_2,\deg(y_1)\right)\right)\\
  \end{aligned}
\end{equation}

since $\alpha_{\varphi} = 1 - \frac{k_1^{k_1}e^{-k_1}}{k_1!}$.

Finally, $\mathbb{E}\left[\sum_{i_2=1}^{\ell_2}\min\left(k_1,\deg_{\mathcal{Y}_1,\mathcal{P}_2}(\mathcal{P}_2^{i_2})\right)\right] = \sum_{i_2=1}^{\ell_2}\mathbb{E}\left[\varphi(\deg_{\mathcal{Y}_1,\mathcal{P}_2}(\mathcal{P}_2^{i_2}))\right]$, so we get that:
\[ \mathbb{E}\left[\sum_{i_2=1}^{\ell_2}\min\left(k_1,\deg_{\mathcal{Y}_1,\mathcal{P}_2}(\mathcal{P}_2^{i_2})\right)\right] \]
is larger than or equal to:
\[ \frac{\ell_2}{k_2}\left(1 - \frac{k_1^{k_1}e^{-k_1}}{k_1!}\right)\left(1-\left(1-\frac{1}{\ell_2}\right)^{k_2}\right)\min\left(k_1k_2,\sum_{y_1}\min\left(k_2,\deg(y_1)\right)\right) \ . \]

Thus, in particular, there exists some partition $\mathcal{P}_2$ that satisfies the same inequality, which concludes the second part of the proof.

\item Let us consider an optimal solution $r_{x,y_1,y_2},p_x,r^1_{x,y_1},r^2_{x,y_2}$ of the program computing $\mathrm{S}^{\textrm{NS}}(W,k_1,k_2)$, so that $\mathrm{S}^{\textrm{NS}}(W,k_1,k_2) = \frac{1}{k_1k_2}\sum_x r_{x,W_1(x),W_2(x)}$.
  \begin{enumerate}
  \item It comes directly from $r_{x,y_1,y_2} \leq p_x$ that:
    \[ \sum_x r_{x,W_1(x),W_2(x)} \leq \sum_x p_x = k_1k_2 \ . \]
  \item $\sum_x r_{x,W_1(x),W_2(x)} = \sum_{y_1}\sum_{x:W_1(x)=y_1} r_{x,y_1,W_2(x)}$ and we have that:
    \begin{enumerate}
    \item $\sum_{x:W_1(x)=y_1} r_{x,y_1,W_2(x)} \leq \sum_{x:W_1(x)=y_1} 1 = \deg(y_1)$ ,
    \item $\sum_{x:W_1(x)=y_1} r_{x,y_1,W_2(x)} \leq \sum_{x:W_1(x)=y_1} r^1_{x,y_1} \leq \sum_x r^1_{x,y_1} = k_2$ ,
    \end{enumerate}
    so $\sum_{x:W_1(x)=y_1} r_{x,y_1,W_2(x)} \leq \min(k_2,\deg(y_1))$, and thus:
    \[ \sum_x r_{x,W_1(x),W_2(x)} \leq \sum_{y_1}\min(k_2,\deg(y_1)) \ .\]
  \end{enumerate}

  In all, we get that:
  \[ \mathrm{S}^{\textrm{NS}}(W,k_1,k_2) = \frac{1}{k_1k_2}\sum_x r_{x,W_1(x),W_2(x)} \leq \frac{\min\left(k_1k_2,\sum_{y_1}\min(k_2,\deg(y_1))\right)}{k_1k_2} \ , \]
  which concludes the third and last part of the proof.
  \end{enumerate}
\end{proof}

\section{Hardness of Approximation for Broadcast Channel Coding}
\label{section:HardnessBC}
Since broadcast channels are more general than point-to-point channels (by defining $W_1(y_1|x):=\hat{W}(y_1|x)$ for $\hat{W}$ a point-to-point channel and taking $W_2(y_2|x)=\frac{1}{|\mathcal{Y}_2|}$ a completely trivial channel), computing a single value $\mathrm{S}(W,k_1,k_2)$ is \textrm{NP}-hard, and it is even \textrm{NP}-hard to approximate within a better factor than $1-e^{-1}$, as a consequence of the hardness result for point-to-point channels from~\cite{BF18}.

The goal of this section is to give some evidence for the hardness of approximation of the general broadcast channel coding problem, specifically that it cannot be approximated in polynomial time within a $\Omega(1)$ factor. This suggests that non-signaling assistance might enlarge the capacity region of the channel as discussed in the introduction.

Formally, one would want to show that it is \textrm{NP}-hard to approximate this problem within a $\Omega(1)$ factor in polynomial time. We were however unable to prove this. Instead, we will prove a $\Omega\left(\frac{1}{\sqrt{m}}\right)$-approximation hardness in the value query model.

First, let us introduce formally the problem:
\begin{definition}[\textsc{BCC}]
  Given a channel $W$ and integers $k_1,k_2$, the broadcast channel coding problem, which we call \textsc{BCC}, entails maximizing:
  \[ \mathrm{S}(W,k_1,k_2,e,d_1,d_2) := \frac{1}{k_1k_2}\sum_{i_1,i_2,y_1,y_2} W(y_1y_2|e(i_1,i_2))\mathbbm{1}_{d_1(y_1)=i_1, d_2(y_2)=i_2} \ ,\]
  over all functions $e : [k_1] \times [k_2] \rightarrow \mathcal{X}$, $d_1 : \mathcal{Y}_1 \rightarrow [k_1]$ and $d_2 : \mathcal{Y}_2 \rightarrow [k_2]$.
\end{definition}

As in the deterministic case, we restrict ourselves to deterministic encoders and decoders, which does not change the value nor the hardness of the problem. Also, it can be equivalently stated in terms of partitions corresponding to $d_1,d_2$ as:
\begin{proposition}[Equivalent formulation of \textsc{BCC}]
  Given a channel $W$ and integers $k_1$ and $k_2$, the broadcast channel coding problem, which we call \textsc{BCC}, entails maximizing:
  \[ \frac{1}{k_1k_2}\sum_{i_1,i_2} \max_x \sum_{y_1 \in \mathcal{P}_1^{i_1}, y_2 \in \mathcal{P}_1^{i_2}} W(y_1y_2|x) \ , \]
  over all partitions $\mathcal{P}_1$ of $\mathcal{Y}_1$ in $k_1$ parts and $\mathcal{P}_2$ of $\mathcal{Y}_2$ in $k_2$ parts.
\end{proposition}

\subsection{Social Welfare Reformulation}
The social welfare maximization problem is defined as follows: given a set $M$ of $m$ items as well as $k$ bidders with their associated utilities $\left(v_i : 2^M \rightarrow \mathbb{R}_+\right)_{i \in [k]}$, the goal is to partition $M$ between the bidders to maximize the sum of their utilities. Formally, we want to compute:
 \[ \underset{\mathcal{P} \text{ partition in $k$ parts of } M}{\maxi}\sum_{i=1}^k v_i\left(\mathcal{P}^i\right) \ . \]

 Let us show that the subproblem of \textsc{BCC} restricted to $k_2=|\mathcal{Y}_2|$ can be reformulated as a particular instance of the social welfare maximization problem. In that case, it is easy to see that $\mathcal{P}_2 = (\{y_2\})_{y_2 \in \mathcal{Y}_2}$ is always an optimal solution. Indeed, for any partition $\mathcal{P}_2$, we have:
\begin{equation}
  \begin{aligned}
    \frac{1}{k_1|\mathcal{Y}_2|}\sum_{i_1,i_2}  \max_x \sum_{y_1 \in \mathcal{P}_1^{i_1}, y_2 \in \mathcal{P}_2^{i_2}} W(y_1y_2|x) &\leq \frac{1}{k_1|\mathcal{Y}_2|}\sum_{i_1,i_2}  \sum_{y_2 \in \mathcal{P}_2^{i_2}} \max_x \sum_{y_1 \in \mathcal{P}_1^{i_1}} W(y_1y_2|x) \\
    &= \frac{1}{k_1|\mathcal{Y}_2|}\sum_{i_1}  \sum_{y_2 \in \mathcal{Y}_2} \max_x \sum_{y_1 \in \mathcal{P}_1^{i_1}} W(y_1y_2|x) \ .
  \end{aligned}
\end{equation}

Therefore, the objective function becomes:
\[ \mathrm{S}^1(W,k_1,\mathcal{P}_1) := \frac{1}{k_1}\sum_{i_1=1}^{k_1} f_W^1(\mathcal{P}_1^{i_1}) \text{ with } f_W^1(S_1) := \frac{1}{|\mathcal{Y}_2|}\sum_{y_2} \max_x \sum_{y_1 \in S_1} W(y_1y_2|x)\ .\]

Hence, up to a multiplicative factor $k_1$, maximizing $\mathrm{S}^1(W,k_1,\mathcal{P}_1)$ over all partitions $\mathcal{P}_1$ of size $k_1$ is a particular case of the social welfare maximization problem with a common utility $f_W^1$ for all $k_1$ bidders.

\subsection{Value Query Hardness}
Let us first introduce the value query model. As described in~\cite{DS06,MSV08}, a value query to a utility $v$ asks for the value of some input set $S \subseteq M$, and gets as response $v(S) \in \mathbb{R}_+$. In the value query model, we aim at solving the social welfare maximization problem accessing the data only through value queries to $(v_i)_{i \in [k]}$.

This is more restricted than using any algorithm, but in such a model, it is possible to show unconditional lower bounds on the number of queries needed to solve a given problem within an approximation rate. In the case of the social welfare maximization problem with \textrm{XOS} utility functions, the approximation rate achievable in polynomial time has been proved in~\cite{DS06,MSV08} to be of the order of $\Theta\left(\frac{1}{\sqrt{m}}\right)$. Specifically, in~\cite{DS06}, a $\Omega\left(\frac{1}{m^{\frac{1}{2}}}\right)$-approximation in polynomial time was given, and in~\cite{MSV08}, it has been shown that any $\Omega\left(\frac{1}{m^{\frac{1}{2}-\varepsilon}}\right)$-approximation for $\varepsilon > 0$ requires an exponential number of value queries. We will adapt their proof in the particular case of one common \textrm{XOS} utility function of the form $f_W^1$ for some broadcast channel $W$. But first, let us introduce the definition of \textrm{XOS} functions and prove that $f_W^1$ is one of those.

\begin{definition}
  A linear valuation function (also known as additive) is a set function $a : 2^M \rightarrow \mathbb{R}_+$ that assigns a nonnegative value to every singleton $\{ j\}$ for $j \in M$, and for all $S \subseteq M$ it holds that $a(S) = \sum_{j \in S} a(\{ j\})$.

  A fractionally sub-additive function (\textrm{XOS}) is a set function $f : 2^M \rightarrow \mathbb{R}_+$, for which there is a finite set of linear valuation functions $A = \{ a_1, \ldots, a_{\ell} \}$ such that $f(S) = \max_{i \in [\ell]} a_i(S)$ for every $S \subseteq M$.
\end{definition}

\begin{rk}
  Note that the size of $A$ is not bounded in the definition. 
\end{rk}

\begin{proposition}
  $f_W^1$ is \textrm{XOS}.
\end{proposition}
\begin{proof}
  \begin{equation}
    \begin{aligned}
      &&f_W^1(S) &= \frac{1}{|\mathcal{Y}_2|}\sum_{y_2} \max_x \sum_{y_1 \in S_1} W(y_1y_2|x) = \max_{\lambda : \mathcal{Y}_2 \rightarrow \mathcal{X}} a_{\lambda}(S) \text{ , where}\\
      &&a_{\lambda}(S) &= \frac{1}{|\mathcal{Y}_2|}\sum_{y_2}  \sum_{y_1 \in S} W(y_1y_2|\lambda(y_2)) = \sum_{y_1 \in S} \left[\frac{1}{|\mathcal{Y}_2|}\sum_{y_2}  W(y_1y_2|m(y_2))\right]\\
      &&&= \sum_{y_1 \in S} a_{\lambda}(\{ y_1 \})) \text{ with } a_{\lambda}(\{ y_1 \})) = \frac{1}{|\mathcal{Y}_2|}\sum_{y_2}  W(y_1y_2|\lambda(y_2)) \in \mathbb{R}_+
      \end{aligned}
  \end{equation}

  So $f_W^1$ is the maximum of the set of $a_{\lambda}$ for $\lambda \in \mathcal{X}^{\mathcal{Y}_1}$, which are linear valuation functions, thus $f_W^1$ is \textrm{XOS}.
\end{proof}

Let us now state the value query hardness of approximation of the broadcast channel problem:
\begin{theorem}
  \label{theo:VQhardnessBC}
  In the value query model, for any fixed $\varepsilon > 0$, a $\Omega\left(\frac{1}{m^{\frac{1}{2}-\varepsilon}}\right)$-approximation algorithm for the broadcast channel coding problem on $W,k_1,k_2$, restricted to the case of $|\mathcal{Y}_2| = k_2$ and $m = |\mathcal{Y}_1| = k_1^2$, requires exponentially many value queries to $f_W^1$.
\end{theorem}

\begin{rk}
  As our problem is a particular instance of the social welfare maximization problem with \textrm{XOS} functions, the polynomial-time $\Omega\left(\frac{1}{m^{\frac{1}{2}}}\right)$-approximation from~\cite{DS06} works also here.
\end{rk}

\begin{proof}
  The proof is inspired by Theorem 3.1 of~\cite{MSV08}. We will show using probabilistic arguments that any $\Omega\left(\frac{1}{m^{\frac{1}{2}-\varepsilon}}\right)$-approximation algorithm requires an exponential number of value queries.
  Let us fix a small constant $\delta > 0$. We choose $k_1 \in \mathbb{N}$ as the number of messages (the bidders) and the output space $\mathcal{Y}_1 := [m]$ with $m := k_1^2$ (the items). Then, we choose uniformly at random an equi-partition of $\mathcal{Y}_1$ in $k_1$ parts of size $k_1$, which we name $T_1, \ldots, T_{k_1}$.

  Let us define now $\mathcal{Y}_2 := [m+k_1+1]$. We take $\mathcal{X} := \mathcal{Y}_2 = [m+1+k_1]$ as well. We can now define our broadcast channel $W$, with some positive constant $C$ to be fixed later to guarantee that $W$ is a conditional probability distribution. Let us define its value for $y_2=1$:
  \[ W(y_1 1|x) := C\times\begin{cases}
    m^{2\delta}\mathbbm{1}_{y_1=x} & \text{when } 1 \leq x \leq m \ ,\\
  \frac{1}{m^{\frac{1}{2}-\delta}} & \text{when } x = m+1 \ ,\\
  \mathbbm{1}_{y_1 \in T_j} & \text{when } 1 \leq j := x-(m+1) \leq k_1 \ .\\
  \end{cases}
  \]

  Then, we define other $y_2$ inputs as translations of $W(y_1 1|x)$. Specifically, we define:
  \[ W(y_1y_2|x) := W(y_11|t_{y_2-1}(x)) \text{ with } t_s(x) := 1 + [(x-1+s) \mod (m+k_1+1)] \ .\] 

  All coefficients are nonnegative. So $W$ will be a channel if for all $x$, $\sum_{y_1,y_2} W(y_1 y_2|x) = 1$. However, one has, for some fixed $x_0$:
  \begin{equation}
    \begin{aligned}
      \sum_{y_1,y_2} W(y_1 y_2|x_0) &= \sum_{y_1} \sum_{y_2} W(y_1 y_2|x_0) = \sum_{y_1} \sum_{y_2}  W(y_1 1|t_{y_2-1}(x_0)) = \sum_{y_1} \sum_x  W(y_1 1|x) \\
      &= C \sum_{y_1}\left[\sum_{1 \leq i \leq m} m^{2\delta}\mathbbm{1}_{y_1=i} + \frac{1}{m^{\frac{1}{2}-\delta}} + \sum_{1 \leq j \leq k_1} \mathbbm{1}_{y_1 \in T_j}\right]\\
      &= C \left[\sum_{1 \leq i \leq m} m^{2\delta} + m\times\frac{1}{m^{\frac{1}{2}-\delta}} + \sum_{1 \leq j \leq k_1} k_1\right]\\
      &= 1 \ ,
    \end{aligned}
  \end{equation}
 by choosing $C = \frac{1}{m^{1+2\delta} + m^{\frac{1}{2}+\delta} + m}$, which does not depend on $x_0$. Thus, we have defined a correct instance of our problem. Note that on this instance, we have:
  \begin{equation}
    \begin{aligned}
      f_W^1(S) &= \frac{1}{|\mathcal{Y}_2|}\sum_{y_2} \max_x \sum_{y_1 \in S} W(y_1y_2|x) = \sum_{y_2} \max_x \sum_{y_1 \in S} W(y_11|t_{y_2-1}(x))\\
      &= \frac{m+k_1+1}{|\mathcal{Y}_2|}\max_x \sum_{y_1 \in S} W(y_11|x) \text{ since $t_{y_2-1}$ bijection}\\
      &= \frac{C(m+k_1+1)}{|\mathcal{Y}_2|} \times \max\begin{cases}
      m^{2\delta}|\{i\} \cap S| & \text{for } 1 \leq i \leq m\\
      \frac{1}{m^{\frac{1}{2}-\delta}}|S|\\
      |T_j \cap S| & \text{for } 1 \leq j \leq k_1\\
      \end{cases}
    \end{aligned}
  \end{equation}
      
  Let us also consider an alternate broadcast channel $W'$, with the only difference that $\mathbbm{1}_{y_1 \in T_j}$ is replaced by $\frac{1}{m^{\frac{1}{2}}}$, for $j \in [k_1]$. For that channel, the constant $C$ remains the same (since $\sum_j\sum_{y_1}\mathbbm{1}_{y_1 \in T_j} = k_1 \times k_1 = k_1 \times m \times \frac{1}{k_1} = \sum_j\sum_{y_1}\frac{1}{m^{\frac{1}{2}}}$), so we get that:
  \begin{equation}
    \begin{aligned}
      f_{W'}^1(S) &= \frac{C(m+k_1+1)}{|\mathcal{Y}_2|} \times \max\begin{cases}
      m^{2\delta}|\{i\} \cap S| & \text{for } 1 \leq i \leq m\\
      \frac{1}{m^{\frac{1}{2}-\delta}}|S|\\
      \frac{1}{m^{\frac{1}{2}}}|S| & \text{for } 1 \leq j \leq k_1\\
      \end{cases}\\
      &= \frac{C(m+k_1+1)}{|\mathcal{Y}_2|} \times \max\begin{cases}
      m^{2\delta}|\{i\} \cap S| & \text{for } 1 \leq i \leq m\\
      \frac{1}{m^{\frac{1}{2}-\delta}}|S|\\
      \end{cases}
    \end{aligned}
  \end{equation}

  since $\frac{1}{m^{\frac{1}{2}}}|S| \leq \frac{1}{m^{\frac{1}{2}-\delta}}|S|$. Let us consider normalized versions $v(S) := \frac{|\mathcal{Y}_2|}{C(m+k_1+1}f_W^1(S)$ and $v'(S) := \frac{|\mathcal{Y}_2|}{C(m+k_1+1}f_{W'}^1(S)$, so distinguishing between $v$ and $v'$ is the same as distinguishing between $f_W^1$ and $f_{W'}^1$. We will prove that it takes an exponential number of value queries to distinguish between $v$ and $v'$. On the one hand, one can easily show that the maximum value of the social welfare problem with $v'$ is $(k_1-1)m^{2\delta}+\frac{1}{m^{\frac{1}{2}-\delta}}(m-(k_1-1)) = O(m^{\frac{1}{2}+2\delta})$, obtained taking $(k_1-1)$ singletons as the first components of the partition (the bidders), giving the rest of $\mathcal{Y}_1$ (the items) to the last. On the other hand, the maximum value of the social welfare problem with $v$ is $k_1 \times k_1 = m$, obtained with the partition $T_1, \ldots, T_{k_1}$. The fact that it requires an exponential number of value queries to distinguish between the two situations will imply that one cannot get an approximation rate better than $\Omega\left(\frac{1}{m^{\frac{1}{2}-2\delta}}\right)$ in less than an exponential number of value queries.

  We will now prove that distinguishing between $v$ and $v'$ requires an exponential number of value queries. Note first that $v(\emptyset) = v'(\emptyset) = 0$, so we do not need to consider empty sets.

  Let us fix some non-empty set $S \subseteq [m]$. Let us define the random boolean variables $X_j^i := \mathbbm{1}_{i \in T_j}$ for $j \in [k_1]$ and $i \in [m]$. By construction of the random equi-partition $T_1, \ldots, T_{k_1}$, $(X_j^i)_{i \in [m]}$ is a permutation distribution (see Definition~\ref{defi:perm}) of $(0,\ldots,0,1,\ldots,1)$ with $m-k_1$ zeros and $k_1$ ones, each $X_j^i$ following a Bernouilli law of parameter $p := \frac{1}{k_1}$. Thus it is negatively associated by Proposition~\ref{prop:permNA}, and the sub-family $(X_j^i)_{i \in S}$ is negatively associated as well by Proposition~\ref{prop:NA_P4}. Note in particular that $|T_j \cap S| = \sum_{i \in S} X_j^i$ is a sum of negatively associated Bernouilli variables of the same parameter $p$, so the version of the Chernoff-Hoeffding bound from Proposition~\ref{prop:chernoff2} holds.
  
  Let us first assume that $S$ is of size $0 < |S| \leq m^{\frac{1}{2}+\delta}$. Then, we have that $\frac{1}{m^{\frac{1}{2}-\delta}}|S| \leq m^{2\delta}$, so we get that $v'(S) = m^{2\delta}$. On the other hand, we have that:
    \begin{equation}
    \begin{aligned}
      v(S) &= \max\begin{cases}
      m^{2\delta}\\
      |T_j \cap S| \text{ for } 1 \leq j \leq k_1\\
      \end{cases}
    \end{aligned}
    \end{equation}

    Thus, $v(S)$ is different from $v'(S)$ if and only if $\exists j \in k_1, |T_j \cap S| > m^{2\delta}$. But, we have:
    \begin{equation}
      \begin{aligned}
        &\mathbb{P}\left(\exists j \in k_1, |T_j \cap S| > m^{2\delta}\right) \leq \sum_{j \in [k_1]} \mathbb{P}\left(|T_j \cap S| > m^{2\delta}\right) \text{ by union bound}\\
        &=\sum_{j \in [k_1]} \mathbb{P}\left(\sum_{i \in S}X_j^i > m^{2\delta}\right)
        =\sum_{j \in [k_1]} \mathbb{P}\left(\frac{1}{|S|}\sum_{i \in S}X_j^i > \left(1 + \frac{\frac{m^{2\delta}}{[S|}-1}{p}\right)p \right)\\
        &\leq \sum_{j \in [k_1]} \mathbb{P}\left(\frac{1}{|S|}\sum_{i \in S}X_j^i > \left(1 + \frac{\frac{m^{2\delta}}{|S|}}{p}\right)p \right)\\
        &\leq \sum_{j \in [k_1]}\exp\left(-\frac{p|S|}{4}\left(\frac{m^{2\delta}}{p|S|}\right)^2\right) \text{ by Proposition~\ref{prop:chernoff2}}\\
        &= \sum_{j \in [k_1]}\exp\left(-\frac{1}{4p|S|}m^{4\delta}\right) \leq \sum_{j \in [k_1]}\exp\left(-\frac{m^{-\delta}}{4}m^{4\delta}\right) \text{ since } \frac{1}{p|S|} = \frac{k_1}{|S|} \geq m^{-\delta}\\
        &= m^{\frac{1}{2}}e^{-\frac{m^{3\delta}}{4}}\ .
      \end{aligned}
    \end{equation}

    Thus, this event occurs with exponentially small probability (on the choice of the partition $T_1,\ldots,T_{k_1}$).

    Let us now study the case of $S$ of size $|S| > m^{\frac{1}{2}+\delta}$. Then, we have that $\frac{1}{m^{\frac{1}{2}-\delta}}|S| > m^{2\delta}$, so we get that $v'(S) = \frac{1}{m^{\frac{1}{2}-\delta}}|S|$. On the other hand, we have that:   
    \begin{equation}
    \begin{aligned}
      v(S) &= \max\begin{cases}
      \frac{1}{m^{\frac{1}{2}-\delta}}|S|\\
      |T_j \cap S| \text{ for } 1 \leq j \leq k_1\\
      \end{cases}
    \end{aligned}
    \end{equation}
    
    Thus, $v(S)$ is different from $v'(S)$ if and only if $\exists j \in [k_1], |T_j \cap S| > \frac{1}{m^{\frac{1}{2}-\delta}}|S|$. But, we have:
    \begin{equation}
      \begin{aligned}
        &\mathbb{P}\left(\exists j \in [k_1], |T_j \cap S| > \frac{1}{m^{\frac{1}{2}-\delta}}|S|\right) \leq \sum_{j \in [k_1]} \mathbb{P}\left(|T_j \cap S| > \frac{1}{m^{\frac{1}{2}-\delta}}|S|\right) \text{ by union bound}\\
        &=\sum_{j \in [k_1]} \mathbb{P}\left(\sum_{i \in S}X_j^i > \frac{1}{m^{\frac{1}{2}-\delta}}|S|\right)
        =\sum_{j \in [k_1]} \mathbb{P}\left(\frac{1}{|S|}\sum_{i \in S}X_j^i > \left(1 + \frac{\frac{1}{|S|m^{\frac{1}{2}-\delta}}{|S|}-1}{p}\right)p \right)\\
        &\leq \sum_{j \in [k_1]} \mathbb{P}\left(\frac{1}{|S|}\sum_{i \in S}X_j^i > \left(1 + \frac{\frac{1}{m^{\frac{1}{2}-\delta}}}{p}\right)p \right)\\
        &= \sum_{j \in [k_1]} \mathbb{P}\left(\frac{1}{|S|}\sum_{i \in S}X_j^i > \left(1 + m^{\delta}\right)p \right) \text{ since } p = \frac{1}{k_1} = \frac{1}{m^{\frac{1}{2}}}\\
        &\leq \sum_{j \in [k_1]}\exp\left(-\frac{p|S|}{4}m^{2\delta}\right) \text{ by Proposition~\ref{prop:chernoff2}}\\
        &\leq \sum_{j \in [k_1]}\exp\left(-\frac{m^{\delta}}{4}m^{2\delta}\right) \text{ since } p|S| = \frac{|S|}{k_1} \geq m^{\delta}\\
        &= m^{\frac{1}{2}}e^{-\frac{m^{3\delta}}{4}}\ .
      \end{aligned}
    \end{equation}

    Thus, this event occurs with exponentially small probability as well. We have then that for all set $S$,  $\mathbb{P}\left(v(S) \not= v'(S) \right) \leq p_{\text{leak}} := m^{\frac{1}{2}}e^{-\frac{m^{3\delta}}{4}}$, which is an exponentially small bound that does not depend on $S$.

    Hence, for every set $S$, only with exponentially small probability $p_{\text{leak}}$ can one distinguish between $v$ and $v'$. For some fixed algorithm $\mathcal{A}$, let us consider the sequence $L$ of queries made by $\mathcal{A}$ before it is able to distinguish between $v$ and $v'$: $L := (S_1, \ldots, S_n)$, with $v(S_i) = v'(S_i)$ for $i \in [n]$ and $v'(S_{n+1}) > v(S_{n+1})$. $L$ is independent of $T_1, \ldots T_{k_1}$ as no information from this partition is leaked before $S_{n+1}$. Thus, for such an algorithm to be correct, it should work for any equi-partition $T_1, \ldots T_{k_1}$. We have:
    \[ \mathbb{P}\left( \exists i \in [n] : v(S_i) \not= v'(S_i) \right) \leq \sum_{i=1}^n\mathbb{P}\left(v(S_i) \not= v'(S_i) \right) = np_{\text{leak}} \text{ by union bound.}\]

    In particular, this implies that:
    \[ \mathbb{P}\left( \forall i \in [n] : v(S_i) = v'(S_i) \right) \geq 1 - np_{\text{leak}} \ . \]

    So, if $1 - np_{\text{leak}} > 0$, i.e. $n < \frac{1}{p_{\text{leak}}}$, then there exists some equi-partition $T_1, \ldots T_{k_1}$ such that our algorithm outputs a sequence $L$ of queries of length $n$ before being able to distinguish between $v$ and $v'$. In particular, we can take $n = \frac{1}{2p_{\text{leak}}}$ so that $L$ is of exponential size. Hence, for any algorithm $\mathcal{A}$, there exists some equi-partition $T_1, \ldots T_{k_1}$ such that $\mathcal{A}$ needs an exponential number of value queries to distinguish between $v$ and $v'$. This concludes the proof of the theorem for any deterministic algorithm.

    Finally, the hardness result holds also for randomized algorithms. Indeed, let us call $\mathcal{A}_s$, the running algorithm conditioned on its random bits being $s$. $\mathcal{A}_s$ is deterministic so the previous proof holds: with high probability $p$, the sequence of $\lfloor\frac{1-p}{p_{leak}}\rfloor$ queries does not reveal anything to distinguish between $v$ and $v'$, although it is of exponential size in $m$. Then, averaging over all its random bitstrings, the same result holds, as $p_{leak}$ is independent of the equi-partition $T_1, \ldots, T_{k_1}$.
\end{proof}

\subsection{Limitations of the Model}
The main weakness of the previous result is that it highly relies on the restriction that one has access to the data only through value queries. Indeed, if one has access to the full data, it is possible to read the partition $T_1, \ldots, T_{k_1}$ which gives the optimal solution directly. This weakness comes from the fact that our utility function $f_W^1$ can be described by polynomial-size data, as it is characterized by a broadcast channel $W$, whereas if we write it in the XOS form as a maximum of linear valuation functions, it will in general have an exponential-size defining set of linear valuation functions.


\section{Conclusion}
We have studied several algorithmic aspects and non-signaling assisted capacity regions of broadcast channels. We have shown that sum success probabilities of the broadcast channel coding problem are the same with and without non-signaling assistance between decoders, and that it implied that non-signaling resource shared between decoders does not change the capacity region. For the class of deterministic broadcast channels, we have described a $(1-e^{-1})^2$-approximation algorithm running in polynomial time, and we have shown that the capacity region for that class is the same with or without non-signaling assistance. Finally, we have shown that in the value query model, we cannot achieve a better approximation ratio than $\Omega\left(\frac{1}{\sqrt{m}}\right)$ in polynomial time for the general broadcast channel coding problem, with $m$ the size of one of the outputs of the channel.

Our results suggest that full non-signaling assistance between the three parties could improve the capacity region of general broadcast channels, which is left as a major open question. An intermediate result would be to show that it is \textrm{NP}-hard to approximate the broadcast channel coding problem within any constant ratio, strengthening our hardness result without relying on the value query model. Finally, one could also try to develop approximations algorithms for other sub-classes of broadcast channels, such as semi-deterministic or degraded ones. This could be a crucial step towards showing that the capacity region for those classes is the same with or without non-signaling assistance.

\section*{Acknowledgements}
We would like to thank Stéphan Thomassé and Siddharth Barman for discussions on graph algorithms and social welfare problems respectively. We would also like to thank Mario Berta and Andreas Winter for discussions on the broadcast channel and the anonymous referees for their constructive feedback. This work is funded by the European Research Council (ERC Grant AlgoQIP, Agreement No. 851716). It has also received funding from the European Union’s Horizon 2020 research and innovation programme within the QuantERA II Programme under Grant Agreement No 101017733.
  
\bibliographystyle{unsrturl}
\bibliography{Broadcast}
\end{document}